\documentclass[onecolumn,journal]{IEEEtran}
\usepackage[table]{xcolor}
\usepackage{tikz}
\usepackage{amsmath,amssymb,amsthm,booktabs,subcaption,arydshln}
\usepackage[inline]{enumitem}
\usepackage{marvosym}
\usepackage{frcursive}
\usepackage[T1]{fontenc}
\usepackage{pgfplots}
\usepackage{soul,mathdots,blkarray,multirow}

\definecolor{DarkGreen}{rgb}{0.1,0.5,0.1}
\definecolor{DarkRed}{rgb}{0.5,0.1,0.1}
\definecolor{DarkBlue}{rgb}{0.1,0.1,0.5}
\usepackage[pdftex]{hyperref}
\hypersetup{
    unicode=false,          %
    pdftoolbar=true,        %
    pdfmenubar=true,        %
    pdffitwindow=false,      %
    pdfnewwindow=true,      %
    colorlinks=true,       %
    linkcolor=DarkBlue,          %
    citecolor=DarkGreen,        %
    filecolor=DarkGreen,      %
    urlcolor=DarkBlue,          %
    pdftitle={},
    pdfauthor={},    
}

\usetikzlibrary{positioning}
\usetikzlibrary{arrows,patterns,calc}

\newtheorem{defn}{Definition}
\newtheorem{lemma}{Lemma}
\newtheorem{thm}{Theorem}
\newtheorem{remark}{Remark}
\newtheorem{corollary}{Corollary}
\newtheorem{construction}{Construction}
\newtheorem{claim}{Claim}
\newtheorem{example}{Example}

\renewcommand{\c}{\ensuremath{\boldsymbol{c}}}

\newcommand{\D}{\ensuremath{\boldsymbol{D}}}

\newcommand{\p}{\ensuremath{\boldsymbol{p}}}
\renewcommand{\P}{\ensuremath{\mathcal{P}}}
\newcommand{\matP}{\ensuremath{\boldsymbol{P}}}
\renewcommand{\u}{\ensuremath{\boldsymbol{u}}}
\newcommand{\x}{\ensuremath{\boldsymbol{x}}}
\newcommand{\y}{\ensuremath{\boldsymbol{y}}}
\newcommand{\Z}{\ensuremath{\mathbb{Z}}}

\newcommand{\Dij}{[\D]_{ij}}

\newcommand{\bintoreal}{\ensuremath{\mathrm{b2r}}}
\newcommand{\ML}{\ensuremath{\mathrm{RV}}}

\newcommand{\ve}[1]{\ensuremath{\boldsymbol{#1}}}
\newcommand{\wt}{\ensuremath{\mathrm{wt}}}
\newcommand{\minmax}{\ensuremath{\mathrm{mm}}}

\newcommand{\Enc}{\ensuremath{\mathsf{Enc}}}
\renewcommand{\Im}{\ensuremath{\mathsf{Im}}}

\newcommand{\I}{E}

\newcommand{\smod}{\ensuremath{~\mathrm{smod}~}}

\def \fcc {FCC}
\def \fccs {FCCs}
\def \FCC {FCC}

\DeclareMathOperator*{\argmax}{arg\,max}
\DeclareMathOperator*{\argmin}{arg\,min}

\newcommand{\PreserveBackslash}[1]{\let\temp=\\#1\let\\=\temp}
\newcolumntype{C}[1]{>{\PreserveBackslash\centering}p{#1}}

\usetikzlibrary{fit}

\usepackage{flushend}

\begin{document}
\title{Function-Correcting Codes}

\author{%
	\IEEEauthorblockN{Andreas Lenz,
		Rawad Bitar,
		Antonia Wachter-Zeh,
		and Eitan Yaakobi}
	\thanks{AL, RB and AW-Z are with the Institute for Communications Engineering, Technical University of Munich (TUM), Germany. Emails: andreas.lenz@mytum.de, \{rawad.bitar, antonia.wachter-zeh\}@tum.de.}
	\thanks{EY is with the CS Department of Technion --- Israel Institute of Technology, Israel. Email: yaakobi@cs.technion.ac.il.}
	\thanks{This project has received funding from the European Research Council (ERC) under the European Union’s Horizon 2020 research and innovation programme (grant agreement No. 801434), from the United States-Israel BSF grant 2018048, and from the Technical University of Munich - Institute for Advanced Studies, funded by the German Excellence Initiative
		and European Union Seventh Framework Programme under Grant Agreement
		No. 291763.
	}
	
\thanks{This paper was presented in part at ISIT 2021 \cite{lenz_function-correcting_2021}.}
}

\markboth{}%
{Lenz \MakeLowercase{\textit{et al.}}: Function-Correcting Codes}

\maketitle

\begin{abstract}
In this paper we study \emph{function-correcting codes}, a new class of codes designed to protect the function evaluation of a message against errors. We show that \fccs~are equivalent to \emph{irregular-distance codes}, i.e., codes that obey some given distance requirement between each pair of codewords. Using these connections, we study irregular-distance codes and derive general upper and lower bounds on their optimal redundancy. Since these bounds heavily depend on the specific function, we provide simplified, suboptimal bounds that are easier to evaluate. We further employ our general results to specific functions of interest and compare our results to standard error-correcting codes, which protect the whole message.
\end{abstract}

\IEEEpeerreviewmaketitle

\section{Introduction}\label{sec:intro}
In standard communication systems, a sender desires to convey a message to a receiver via an erroneous channel. %
Classically, each part of the message is of equal importance to the receiver and the common goal is to construct an error-correcting code with a suitable decoder such that the whole message can be recovered correctly.
Consider now the scenario where a certain \emph{attribute} of the message, i.e., the result of evaluating a certain function on the message, is of particular interest to the receiver. Assuming that the sender is aware of this function, she can encode the message such that the desired attribute is protected against errors. This paradigm gives rise to a new class of codes, which we call \emph{function-correcting codes} (FCCs). %
In this work we consider FCCs, where the message itself is observed through the channel, followed by redundancy, as illustrated in Fig.~\ref{fig:intro_fcc_problem}. 

\begin{figure}[t]
	\centering
	\begin{tikzpicture}[>=stealth]
		\node (u) {\u};
		\node[right = 0.5cm of u, rectangle, rounded corners, fill=gray!20!, draw, minimum height=0.75cm] (enc) {Encoder};
		\node[above left = 0.65cm and -1.05cm of enc] (alice) {\bfseries\large{Alice}};

		\node[draw,rectangle, rounded corners, fill=gray!20!, right = 2.5 cm of enc.south, minimum height=1.5cm, anchor = south] (channel) {Channel};
		\node[draw,rectangle, rounded corners, fill=gray!20!, right = 1 cm of channel, minimum height=0.75cm] (decoder) {Decoder};
		
		\node[below = 0.5cm of channel] (fu) {$f$};
		
		\node[right = 0.5cm of decoder](fy) {$f(\u)$};
		\node[above right= 0.25cm and -0.5cm of decoder] (bob) {\bfseries\large{Bob}};

		\draw[->] (u) -- (enc);%
		\coordinate[] (cor) at ($(u) + (0.5,0.7)$);
		\draw[-] (u -| cor) -- (cor);
		\draw[->] (cor) -- (cor -| channel.west);
		\draw[->] (enc.east) -- node [above] {$\p$} (enc.east -| channel.west);
		\draw[->] (channel) -- node [above] {$\y$} (decoder);
		\draw[->] (decoder) -- (fy);
		\draw[->, dashed] (fu) -| (enc);
		\draw[->, dashed] (fu) -| (decoder);
	\end{tikzpicture}

	\caption{Illustration of the \fcc\ setup. Alice has a message $\u$, which features an attribute $f(\u)$ that is of special interest to Bob. To guarantee recoverability of this attribute, Alice encodes the message $\u$ to a redundancy vector $\p$. Given an erroneous version $\y$ of the codeword $\c=(\u,\p)$ and the knowledge of the function $f$, Bob can correctly infer $f(\u)$.}
	\label{fig:intro_fcc_problem}
\end{figure}%

Clearly, if the receiver is able to recover the message, it can evaluate the function on the message to obtain the desired attribute.
It is however more efficient to protect only the specific function value of interest, especially when the message is long and the function image is small.
Our generic goal when designing \fccs~for a given function is to use the smallest amount of redundancy that allows the recovery of the attribute.%

A key aspect for the design of \fccs~is the topology of the function regions, i.e., the sets of message vectors that evaluate to the same function value. Since \fccs~protect a specific function evaluation of the message, the receiver does not need to distinguish between codewords from messages that evaluate to the same function value. This means that the distance between any two codewords within one function region is irrelevant. On the other hand, codewords corresponding to different function values should have appropriate distances. Consequently, the redundancy vectors of an \fcc~have to fulfill an irregular distance profile, where each pair of redundancy vectors has to satisfy an individual distance constraint. %

\emph{Application:} %
The employment of \fccs~can, for example, be beneficial in archival data storage. Consider a large data set, or message in our terminology, which is stored on a noisy storage medium. The message may be encoded with an error-correcting code to ensure reliability when retrieving it. Now assume that an attribute with highly sensitive information, which can be modeled as a function evaluation on the message is to be stored on the same medium. Due to the importance of this peculiar attribute, we desire to add an extra layer of protection for it. A natural solution is to encode the attribute with an error-correcting code with high error-correction capability. However, this idea is oblivious to the fact that the message is stored on the same storage medium. We propose to leverage this fact through the use of \fccs\ that may require less redundancy as we show in the sequel. When reading the message and its attribute through a noisy channel, errors may happen. To abstract the fact that more errors can happen in the message than the error-correcting code used to encode it can tolerate, we assume that a potentially noisy version of the data is available to the receiver. \fccs~thus provide an individual level of protection to specific attributes of the message, offering higher flexibility and efficiency over classical error-correcting codes.

\emph{Related works:}
Unequal error protection (UEP) codes \cite{masnick_linear_1967,boyarinov_linear_1981} allow a stronger protection of specific parts of the message. %
The connection between \fccs~and codes for UEP manifests in the two levels of protection that \fccs~provide to the message and the attribute.
Since the message can either be unprotected or can itself be the codeword of an error-correcting code and the attribute is separately protected by the \fcc, it is possible to control the error protection level of these parts. As the attribute is the evaluation of an arbitrary function on the message, it is also possible that it is simply a substring of the message, resulting in UEP for this part of the message. In general however, the attribute may be an arbitrary function of the message and hence, in this aspect, \fccs~form a more general class of codes. On the other hand, UEP codes may protect an arbitrary number of attributes of interest.

In the context of random access memories, codes with unequal message protection have been designed in \cite{schoeny_context-aware_2019}. The authors construct codes that guarantee larger distances for codewords that stem from a specific predefined subset of messages. Similar to our work, the required distance varies between pairs of codewords. In contrast to our work, the distance requirement in \cite{schoeny_context-aware_2019} depends on groups to which the codewords are assigned, and here, as we will show later, we require an individual distance for each pair of codewords. For an information-theoretic study of both unequal error and message protection codes see~\cite{borade_unequal_2009}.

Another related line of work \cite{ahlswede_get_1981,orlitsky_coding_2001,kuzuoka_distributed_2016} studies the scenario, where a sender wishes to communicate a message to a receiver such that the receiver can determine the evaluation of a function on their combined data. Therein, optimal transmission rates for which the recovery of the function evaluation is possible, are derived. There are several important aspects that differentiates our work from these papers. First, we study zero-error codes over an adversarial channel as opposed to non-zero, but vanishing error probabilities. Second, the message sent from Alice to Bob may contain errors in our setup. Notice however that \cite{orlitsky_coding_2001} uses a characteristic graph \cite{witsenhausen_zero-error_1976}, which is similar in spirit to the irregular-distance codes defined later.

The zero-rate threshold for adversarial channels is derived in \cite{wang_when_2019}. Studying general channels, \cite{wang_when_2019} deals with a broad notion of \emph{confusability} between codewords. This is similar to the irregular distance codes in this work, however with the important distinction that here the confusability depends on the assignment between message vectors and codewords.

Codes that protect the output of a given machine learning algorithm against errors have been investigated in~\cite{mazooji_robust_2016,kabir_coded_2017,huang_functional_2020-1}. While \cite{mazooji_robust_2016,kabir_coded_2017} tailored their construction to optimize classification algorithms, \cite{huang_functional_2020-1} applied codes to the weights of the neurons in a neural network with the goal to optimize the output model of the neural network. These works prove that application-specific codes that protect the output of an algorithm against errors can outperform classical error-correcting codes. In principle, we follow a similar idea in this work, however the research in \cite{mazooji_robust_2016,kabir_coded_2017,huang_functional_2020-1} is specialized to specific classes of functions, while we discuss arbitrary functions. On the other hand, with the current state of research, it seems infeasible to practically and efficiently apply our generic results to such intricate functions.

We further would like to highlight the following works on error-correction within computations. Fourier stabilization has been used in~\cite{raviv_enhancing_2021} to increase the robustness of a neural network. Therein, error resilience was achieved by replacing the weights of neurons according to the solution of an associated combinatorial optimization problem. In~\cite{roth_fault-tolerant_2019,roth_analog_2020,dupraz_noisy_2020} computation in faulty dot-product engines is treated. While \cite{roth_fault-tolerant_2019,roth_analog_2020} construct codes over integers and real numbers that protect the computation of a matrix-vector product, \cite{dupraz_noisy_2020} propose a theoretical framework for the error analysis of memristor crossbars.
Codes that correct and detect errors in arithmetic operations are discussed in~\cite[ch. 10]{lint_introduction_1999}.

\emph{Contributions:} This paper builds a general theory for function-correction over adversarial channels. For arbitrary functions, we establish a connection between \fccs~and \emph{irregular-distance codes}. In particular, we show that the redundancy of an \FCC~is given by the shortest length of an irregular-distance code, which has a punctured pair-wise distance profile, which depends on the function. Deriving general lower and upper bounds on the optimal length of irregular distance codes, we obtain corresponding bounds on the optimal redundancy of \fccs~for arbitrary functions. These results are applied to specific functions such as locally binary functions, the Hamming weight, the Hamming weight distribution, the min-max function and %
discretized real-valued functions. Finally, the redundancy of \fccs~for specific functions is compared to schemes that use standard error-correcting codes. We
restrict our attention to binary channels in this work,
however most results can be generalized straightforwardly to
larger alphabets. A summary of our quantitative results for specific functions is summarized and displayed in Table \ref{tab:redundancy}.

\emph{Organization:} Section~\ref{sec:prelim} summarizes the problem considered and the main notations of the paper. Next, we study generic functions in Section~\ref{sec:generic} and reveal the fundamental connection between \fccs, irregular-distance codes, and independent sets in certain graphs. To this end, we show that the optimal redundancy of an \fcc~is equal to the smallest length of an irregular-distance code. We then provide simplified results that are easier to evaluate, especially for functions with entwined function regions. Further, generic converse and existential bounds on irregular-distance codes are presented. We then apply our generic results to specific functions in Sections~\ref{sec:locally_binary}, \ref{sec:hamming}, \ref{sec:min_max} and~\ref{sec:ml_functions}. Section~\ref{sec:conclusion} concludes the paper.

\section{Preliminaries}\label{sec:prelim}
\begin{table*}
	\centering
	\caption{Summary of results on the optimal redundancy of \fccs. The entries marked with superscript $*$ are approximations for large dataset dimensions $k$ and expressiveness $\I$ (where applicable), and fixed number of errors $t$, where lower order terms are neglected. The redundancy of \fccs~is displayed for the case where Hadamard matrices of correct size exist, cf.  Lemma \ref{lemma:regular:distance:hadamard}. These restrictions and regimes are chosen to allow for better comparison, however our results are not restricted to these regimes. Precise definitions of the displayed functions can be found in Section \ref{sec:locally_binary} (binary and locally binary), Section \ref{sec:hamming_w} (Hamming weight), Section \ref{sec:hamming_d} (Hamming weight distribution) and Section \ref{sec:min_max} (min-max). The redundancies \emph{ECC on Data} and \emph{ECC on Function Values} are derived in Appendix \ref{app:derivation:redundancies:table}.}
	{\renewcommand{\arraystretch}{2.1}
			\setlength{\tabcolsep}{6.5pt}
		\begin{tabular}{cccccc} \specialrule{.8pt}{0pt}{0pt}
			Function & Parameters & Lower Bound & ECC on Data & ECC on Function Values & \fcc  \\ \specialrule{.8pt}{0pt}{0pt}
			Binary& - & $2t$ & $t \log k$ $~^{*}$ & $2t+1$ & $2t$ \\
			Locally binary & $E$ & $2t$ & $t \log k$ $~^{*}$ & $\log E + t\log\log E$ $~^{*}$ & $2t$ \\
			Hamming weight $\wt(\u)$ & - & $\frac{10}{3}(t-1)$ &  $t \log k$ $~^{*}$ & $\log k + t \log \log k$ $~^{*}$ & $4t$ \\
			Hamming weight distribution $\Delta_T(\u)$ & \parbox{2.7cm}{\centering Threshold $T\geq 2t+1$,\\$E=\frac{k+1}{T}$} & $2t$ & $t \log k$ $~^{*}$ & $\log E + t \log \log E$$~^{*}$ & $2t$ \\
			Min-max $\minmax_w(\u)$  & Num. parts $w \gg 2t$ & $\begin{array}{cc}
				2 \log w ~+ \\[-.3cm]
				(t-2)\log \log w
			\end{array}$ $\!\!\!\!\!\! {\scriptstyle*}$ & $t \log k$ $~^{*}$ & $2 \log w + t \log \log w$ $~^{*}$ & $\begin{array}{cc}
			2 \log w ~+ \\[-.3cm]
			t\log \log w
		\end{array}$ $\!\!\!\! {\scriptstyle*}$ \\ \specialrule{.8pt}{0pt}{0pt}
	\end{tabular}}
	\label{tab:redundancy}
\end{table*}

Let $\u \in \Z_2^k $ be the binary message and let $f: \Z_2^k \mapsto \Im(f) \triangleq \{f(\u): \u \in \Z_2^k\}$ be a function computed on $\u$ with \emph{expressiveness} $\I\triangleq |\Im(f)|\leq 2^k$.\footnote{The nature of the image $\Im(f)$, apart from its size, is not relevant in this paper. Thus, it is not further specified.} The message is encoded via the encoding function $\Enc:\Z_2^k \mapsto \Z_2^{k+r}$, $ \Enc(\u) = (\u,\p(\u)), $
where $\p(\u) \in \Z_2^r$ is the \emph{redundancy vector} and $r$ is the \emph{redundancy}. The resulting codeword $\Enc(\u)$ is transmitted over an erroneous channel, resulting in $\ve{y} \in \Z_2^{k+r}$ with $d( \Enc(\u),\y) \leq t $, where $d(\x,\y)$ is the Hamming distance of $\x$ and $\y$. We define {\fccs} as follows.
	\begin{defn} \label{def:function:correcting:code}
		An encoding function $\Enc:\Z_2^k\to\Z_2^{k+r}$ with $\Enc(\u) = (\u,\p(\u))$, $\u \in \Z_2^k$ defines a function-correcting code for the function $f:\Z_2^k \to \Im(f)$ if for all $\u_1,\u_2 \in \Z_2^k$ with $f(\u_1) \neq f(\u_2)$, it holds that
		$$ d(\Enc(\u_1), \Enc(\u_2)) \geq 2t+1. $$
	\end{defn}

	By this definition, given any $\y$, which is obtained by at most $t$ errors from $\Enc(\u)$, the receiver can uniquely recover $f(\u)$, if it has knowledge about the function $f(\bullet)$ and the encoding function $\Enc(\bullet)$. Noteworthily, only codewords that originate from information vectors (messages) that evaluate to different function values need to have distance at least $2t+1$. Throughout the paper, a \emph{standard error-correcting code} is an \fcc~for $f(\u)=\u$, i.e., a code that allows to reconstruct the whole message $\u$. We summarize some basic properties of \fccs~in the following.
		\begin{itemize}
			\item For any bijective function $f$, any {\fcc} is a standard error-correcting code.
			\item For any constant function $f$, the encoder $\Enc(\u) = \u$ is an \fcc\ with redundancy $0$.
			\item If the encoder has no knowledge about the function $f$, function-correction is only possible using standard error-correcting codes.
		\end{itemize}

Note that the encoding and decoding complexity of \fccs~may be higher or lower than that of standard error-correcting codes and heavily depends on the function $f$.

The main quantity of interest in this paper is the optimal redundancy of an \fcc~that is designed for a function~$f$.
\begin{defn}
	The optimal redundancy $r_{f}(k,t)$ is defined as the smallest $r$ such that there exists an \fcc\ with encoding function $\Enc:\Z_2^k \to \Z_2^{k+r}$ for the function $f$.
\end{defn}
For any integer $M$, we write $[M]^+ \triangleq\max\{M,0\}$ and we let $[M]\triangleq \{1,\dots,M\}$. For a matrix $\D$, we denote by $[\D]_{ij}$ the $(i,j)$th entry of $\D$. For any two real numbers $a,b\in \mathbb{R}$, we define the closed and half-closed interval by $[a,b] \triangleq \{x \in \mathbb{R}: a\leq x\leq b\}$ and $[a,b) \triangleq \{x \in \mathbb{R}: a\leq x< b\}$. We denote by $\mathbb{N}_0$ the set of non-negative integers. Note that while our quantitative results in this paper are for substitution channels, the concepts can be generalized to other channels.

\section{Generic Functions}\label{sec:generic}

This section is devoted to establishing general results on \fccs. We start by showing the equivalence of \fccs, irregular-distance codes (Definition~\ref{def:fcc2}), and independent sets\footnote{An independent set of an undirected graph is a subset of vertices, where no two vertices are connected by an edge.} in an associated graph (Definition~\ref{def:graph}). We proceed afterwards with establishing several lower and upper bounds on the optimal redundancy of \fccs~using these connections.

We begin with introducing irregular-distance codes. To this end, define the distance matrix of a function $f$ as follows. 
\begin{defn}
    Let $\u_1,\dots,\u_M \in \Z_2^k$. We define the distance requirement matrix $\D_f(t,\u_1,\dots,\u_{M})$ of a function $f$ as the $M \times M$ matrix with entries
	$$ [\D_f\hspace{-0.25ex}(t,\hspace{-0.25ex}\u_1,\hspace{-0.25ex}\dots,\hspace{-0.25ex}\u_{M})]_{ij} \hspace{-0.25ex}\!=\!\hspace{-0.25ex} \left\{ \!\!\! \begin{array}{ll}
	[2t\!+\!1\!-\!d(\u_i,\hspace{-0.25ex}\u_j)]^+, \!\!\!\!\!& \text{\hspace{-0.25ex}if } f(\u_i) \hspace{-0.4ex} \neq\hspace{-0.4ex} f(\u_j),\\
	0,& \text{otherwise.}
	\end{array} \right. \!\!\! $$
\end{defn}
Let $\P = \{\p_1,\p_2,\dots,\p_M\} \subseteq \Z_2^r$ be a code of length $r$ and cardinality $M$. Here, we choose $r$ as the code blocklength, as we will relate the code length $r$ to the redundancy of {\fccs} later. Irregular-distance codes are formally defined as follows. %
\begin{defn}\label{def:fcc2}
	Let $\D \in \mathbb{N}_0^{M\times M}$. Then, $\P = \{\p_1,\p_2,\dots,\p_M\}$ is a $\D$-code, if there exists an ordering of the codewords of $\P$ such that $d(\p_i,\p_j) \geq [\D]_{ij}$ for all $i,j \in [M]$.
	
	Further, we define $N(\D)$ to be the smallest integer $r$ such that there exists a $\D$-code of length $r$. If $\Dij = D$ for all $i \neq j$ we write $N(M, D)$.
\end{defn}

With this definition, a $\D$-code requires individual distances between each pair of codewords.

Next, we define a function-dependent graph, whose independent sets, if large enough, form an \FCC. The vertices constitute possible codewords of the \fcc~and we connect two vertices, if they can be contained together in an \fcc. 
\begin{defn} \label{def:graph}
	We define $G_f(k,t,r)$ to be the graph with vertex set $V = \{0,1\}^k \times \{0,1\}^r$, such that each vertex has the form $\boldsymbol{x} = (\boldsymbol{u},\boldsymbol{p}) \in \{0,1\}^{k+r}$. Two vertices $\boldsymbol{x}_1 = (\boldsymbol{u}_1,\boldsymbol{p}_1)$ and $\boldsymbol{x}_2 = (\boldsymbol{u}_2,\boldsymbol{p}_2)$ are connected if $\u_1=\u_2$ or both $f(\boldsymbol{u}_1) \neq f(\boldsymbol{u}_2)$ and $d(\boldsymbol{x}_1, \boldsymbol{x}_2) < 2t+1$ hold.
	
	We denote by $\gamma_f(k,t)$ the smallest integer $r$ such that there exists an independent set of size $2^k$ in $G_f(k,t,r)$.
\end{defn}
\begin{figure}
	\centering
\pgfdeclarelayer{bg}
\pgfsetlayers{bg,main}
	\begin{tikzpicture}
		\draw (1*360/16: 3.8cm) node[ultra thick, draw, rectangle, fill=black!10!white] (u0000) { \footnotesize $(00,00)$};
		\draw (2*360/16: 3.8cm) node[draw, rectangle, fill=black!10!white] (u0010) { \footnotesize $(00,10)$};
		\draw (3*360/16: 3.8cm) node[draw, rectangle, fill=black!10!white] (u0001) { \footnotesize $(00,01)$};
		\draw (4*360/16: 3.8cm) node[draw, rectangle, fill=black!10!white] (u0011) { \footnotesize $(00,11)$};
		\draw (5*360/16: 3.8cm) node[draw, rectangle, fill=white!10!white] (u1000) { \footnotesize $(10,00)$};
		\draw (6*360/16: 3.8cm) node[draw, rectangle, fill=white!10!white] (u1010) { \footnotesize $(10,10)$};
		\draw (7*360/16: 3.8cm) node[draw, rectangle, fill=white!10!white] (u1001) { \footnotesize $(10,01)$};
		\draw (8*360/16: 3.8cm) node[ultra thick, draw, rectangle, fill=white!10!white] (u1011) { \footnotesize $(10,11)$};
		\draw (9*360/16: 3.8cm) node[draw, rectangle, fill=white!10!white] (u0100) { \footnotesize $(01,00)$};
		\draw (10*360/16: 3.8cm) node[draw, rectangle, fill=white!10!white] (u0110) { \footnotesize $(01,10)$};
		\draw (11*360/16: 3.8cm) node[draw, rectangle, fill=white!10!white] (u0101) { \footnotesize $(01,01)$};
		\draw (12*360/16: 3.8cm) node[ultra thick, draw, rectangle, fill=white!10!white] (u0111) { \footnotesize $(01,11)$};
		\draw (13*360/16: 3.8cm) node[draw, rectangle, fill=white!10!white] (u1100) { \footnotesize $(11,00)$};
		\draw (14*360/16: 3.8cm) node[draw, rectangle, fill=white!10!white] (u1110) { \footnotesize $(11,10)$};
		\draw (15*360/16: 3.8cm) node[draw, rectangle, fill=white!10!white] (u1101) { \footnotesize $(11,01)$};
		\draw (16*360/16: 3.8cm) node[ultra thick, draw, rectangle, fill=white!10!white] (u1111) { \footnotesize $(11,11)$};
		
		\begin{pgfonlayer}{bg}
			\draw[-] (u0000) -- (u0010);
			\draw[-] (u0000) -- (u0001);
			\draw[-] (u0000) -- (u0011);
			\draw[-] (u0010) -- (u0001);
			\draw[-] (u0010) -- (u0011);
			\draw[-] (u0001) -- (u0011);
			
			\draw[-] (u1000) -- (u1010);
			\draw[-] (u1000) -- (u1001);
			\draw[-] (u1000) -- (u1011);
			\draw[-] (u1010) -- (u1001);
			\draw[-] (u1010) -- (u1011);
			\draw[-] (u1001) -- (u1011);
			
			\draw[-] (u0100) -- (u0110);
			\draw[-] (u0100) -- (u0101);
			\draw[-] (u0100) -- (u0111);
			\draw[-] (u0110) -- (u0101);
			\draw[-] (u0110) -- (u0111);
			\draw[-] (u0101) -- (u0111);
			
			\draw[-] (u1100) -- (u1110);
			\draw[-] (u1100) -- (u1101);
			\draw[-] (u1100) -- (u1111);
			\draw[-] (u1110) -- (u1101);
			\draw[-] (u1110) -- (u1111);
			\draw[-] (u1101) -- (u1111);
			
			\draw[-] (u0011) -- (u1010);
			\draw[-] (u0011) -- (u1001);
			\draw[-] (u0011) -- (u1011);
			\draw[-] (u0011) -- (u0110);
			\draw[-] (u0011) -- (u0101);
			\draw[-] (u0011) -- (u0111);
			\draw[-] (u0011) -- (u1111);
			
			\draw[-] (u0001) -- (u1000);
			\draw[-] (u0001) -- (u1011);
			\draw[-] (u0001) -- (u1001);
			\draw[-] (u0001) -- (u0100);
			\draw[-] (u0001) -- (u0111);
			\draw[-] (u0001) -- (u0101);
			\draw[-] (u0001) -- (u1101);
			
			\draw[-] (u0010) -- (u1011);
			\draw[-] (u0010) -- (u1000);
			\draw[-] (u0010) -- (u1010);
			\draw[-] (u0010) -- (u0111);
			\draw[-] (u0010) -- (u0100);
			\draw[-] (u0010) -- (u0110);
			\draw[-] (u0010) -- (u1110);
			
			\draw[-] (u0000) -- (u1001);
			\draw[-] (u0000) -- (u1010);
			\draw[-] (u0000) -- (u1000);
			\draw[-] (u0000) -- (u0101);
			\draw[-] (u0000) -- (u0110);
			\draw[-] (u0000) -- (u0100);
			\draw[-] (u0000) -- (u1100);

		\end{pgfonlayer}
	\end{tikzpicture}
\caption{Graph $G_f(k,t,r)$ for $k=2,t=1,r=2$, and the function $f((u_1,u_2)) = (u_1 \lor u_2)$. An independent set of size $2^k=4$, i.e., an FCC for $f$, is highlighted in bold. The background colors highlight different function values. }
\label{fig:independent:graph}
\end{figure}
This graph resembles the characteristic graph in \cite{witsenhausen_zero-error_1976}, however differs due to the fact that $\u$ is observed through the channel and that functions depend on the whole message vector in our problem formulation. Note that the edges between vertices with $\u_1=\u_2$ enforce the property that each information vector $\u$ is assigned exactly one redundancy vector $\p(\u)$. Fig.~\ref{fig:independent:graph} visualizes the graph $G_f(k,t,r)$ and a corresponding \fcc~for a concrete example.

We find the following central connection between the redundancy of optimal {\fccs}, irregular-distance codes, and independent sets in the associated graphs..
\begin{thm}[] \label{thm:optimal:fcc:sub}
	For any function $f: \mathbb{Z}_2^k \to \Im(f)$,
	$$ r_f(k,t)= \gamma_f(k,t)= N(\D_f(t,\u_1,\dots,\u_{2^k})), $$
	where $\{\u_1,\dots,\u_{2^k}\} = \Z_2^k$ are all binary vectors of length $k$.
\end{thm}
\begin{proof}
	The first equality is immediate as an independent set in $G_f(k,t,r)$ exactly captures the required properties of an \fcc. Further, the independent set has to have size $2^k$ such that there is one codeword for every message vector.
	
	Next, we see that \mbox{$r_f(k,t)\geq N(\D_f(t,\u_1,\dots,\u_{2^k}))$} is necessary, as assuming to the contrary that $r_f(k,t)<N(\D_f(t,\u_1,\dots,\u_{2^k}))$ implies that there must exist two redundancy vectors $\p_i$ and $\p_j$, $i\neq j$ with $d(\p_i,\p_j)<2t+1-d(\u_i,\u_j)$ and hence $d(\Enc(\u_i),\Enc(\u_j)) = d(\u_i,\u_j) + d(\p_i,\p_j) < 2t+1$, which contradicts Definition~\ref{def:function:correcting:code}.
	
	On the other hand $r_f(k,t)\leq N(\D_f(t,\u_1,\dots,\u_{2^k}))$, as using a correctly assigned $\D_f(t,\u_1,\dots,\u_{2^k})$-code for the redundancy vectors gives an \fcc.
\end{proof}
\begin{remark}
	The irregularity of the distance profile of \fccs\ comes from imposing distance constraints on the redundancy vectors as opposed to codewords. In our analysis, we found this approach to naturally capture the interplay between the message and the redundancy part and to help with the derivation of simplified bounds and constructions, which are presented in the sequel.
\end{remark}

With the result of Theorem~\ref{thm:optimal:fcc:sub}, one can deduce insights into \fccs~using known results about the sizes of independent sets in general graphs, such as \cite{gu_generalized_1993,tolhuizen_generalized_1997}.

However, the problem of finding \emph{optimal} \fccs~requires the determination of whether the size of the largest independent set meets the threshold $2^k$. The related problem of finding a maximal independent set in arbitrary graphs is known to be NP-complete \cite{miller_reducibility_1972}, which indicates that also the problem of finding optimal \fccs~is complex, unless the structure of the analyzed function $f$ imposes an easily tractable graph structure that simplifies the search for large independent sets.

This implies that the construction of optimal \fccs~may become computationally infeasible for large parameters and unstructured functions. To cope with such scenarios, we derive simplified, possibly sub-optimal, results on irregular-distance codes, in order to facilitate the research for arbitrary functions. We proceed with deriving results that act on a smaller set of information vectors and ease the derivation of analytical results.

\subsection{Simplified Redundancy Lower Bounds}
We first compute simplified lower bounds on the optimal redundancy of \fccs. Using an arbitrary subset of information vectors $\u_1,\dots,\u_M$ with $M \leq 2^k$ we can obtain a lower bound on the redundancy as follows.
\begin{corollary} \label{cor:lower:bound}
	Let $\u_1,\dots,\u_{M} \in \Z_2^k$ be arbitrary different vectors. Then, the redundancy of an \fcc\ is at least
	$$r_f(k,t) \geq N(\D_f(t,\u_1,\dots,\u_{M})).$$
	For any function $f$ with $|\Im(f)|\geq2$,
	$$r_f(k,t)\geq 2t.$$
\end{corollary}
\begin{proof}
	The first statement is immediate, since any subset of information vectors must also fulfill the \fcc~conditions.
	
	Since $|\Im(f)| \geq 2$, there exist $\u,\u' \in \Z_2^k$ with $d(\u,\u') = 1$ and $f(\u)\neq f(\u')$. It follows that $r_f(k,t) \geq N(2,2t)$. Further, $N(2,2t)=2t$, which is attained by the repetition code $\P = \{(0,\dots,0), (1,\dots,1)\}$ of length $2t$.
\end{proof}
Finding $N(\D_f(t,\u_1,\dots,\u_{2^k}))$ is in general quite difficult and it can be easier to focus only on a small but representative subset of information vectors. However, the particular subset heavily depends on the function itself and it is not possible to give a generic approach on how a good subset can be found. Loosely speaking, good bounds are obtained for information vectors that have distinct function values and are close in Hamming distance. Throughout this paper, we will provide some insights on good choices of information vectors using illustrative examples.

\subsection{Simplified Existential Bounds}
We proceed with simplifying Theorem~\ref{thm:optimal:fcc:sub} in order to obtain easier computable existential bounds. We start by defining the distance between two function values.
\begin{defn}\label{def:function_distance}
	The distance between two function values $f_1,f_2 \in \Im(f)$ is defined as the smallest distance between two information vectors that evaluate to $f_1$ and $f_2$, i.e.,
	$$ d_f(f_1,f_2) \triangleq \underset{\u_1,\u_2 \in \Z_2^k}{\min} d(\u_1,\u_2)~\text{s.t.}~ f(\u_1) = f_1 \land f(\u_2) = f_2.$$
\end{defn}
Note that the distance $d_f(f_1,f_1) = 0, \forall f_1 \in \Im(f)$. The function-distance matrix of $f$ is thus defined as follows.
\begin{defn}
    The function-distance matrix of a function $f$ is denoted by the $\I \times \I$ matrix $\D_f(t,f_1,\dots,f_\I)$ with entries
    $ [\D_f(t,f_1,\dots,f_\I)]_{ij} = [2t+1-d_f(f_i,f_j)]^+,$
    if $i\neq j$ and $[\D_f(t,f_1,\dots,f_\I)]_{ii} = 0$.
\end{defn}
One way to construct {\fccs} is to assign the same redundancy vector to all information vectors $\u$ that evaluate to the same function value. This is not a necessity, however it gives rise to the following existence theorem.
\begin{thm} \label{thm:achievability}
	For any arbitrary function $f: \mathbb{Z}_2^k \to \Im(f)$,%
	$$ r_f(k,t) \leq N(\D_f(t,f_1,\dots,f_\I)) . $$
\end{thm}
\begin{proof}
	We describe how to construct an \fcc. The redundancy vectors are chosen to depend only on the function value of $\u$, i.e., the encoding mapping is defined by $ \u \mapsto (\u,\p(f(\u)))$. Denote by $\p_i$ the redundancy vector assigned to all $\u$ with $f(\u) = f_i$. Therefore, two information vectors with the same function value have the same redundancy vectors. We then choose $\p_1,\dots,\p_\I$ such that $d(\p_i,\p_j) \geq 2t+1-d_f(f_i,f_j)$. It follows that for any $\u_i,\u_j$ with $f(\u_i) = f_i$, $f(\u_j)=f_j$, $f_i\neq f_j$, we have $d(\Enc(\u_i),\Enc(\u_j)) = d(\u_i,\u_j) + d(\p_i,\p_j) \geq d_f(f_i,f_j) + 2t+1-d_f(f_i,f_j) = 2t+1 $. By Definition~\ref{def:fcc2} we can guarantee the existence of such parity vectors $\p_1,\dots,\p_\I$, if they have length $N(\D_f(t,f_1,\dots,f_\I))$.
\end{proof}
There are cases in which the bound in Theorem \ref{thm:achievability} is tight. We characterize one important case in the following corollary, which is a consequence of Corollary~\ref{cor:lower:bound} and Theorem~\ref{thm:achievability}.
\begin{corollary}
	If there exists a set of representative information vectors $\u_1,\dots,\u_\I$ with $\{f(\u_1),\dots,f(\u_\I)\} = \Im(f)$ and $\D_f(t,\u_1,\dots,\u_\I) = \D_f(t,f_1,\dots,f_\I)$, then
	$$ r_f(k,t) = N(\D_f(t,f_1,\dots,f_\I)). $$
\end{corollary}
Even though the bound in Theorem \ref{thm:achievability} is not necessarily tight, in many cases it is much easier to derive the function distance matrix $\D_f(t,f_1,\dots,f_\I)$ than the distance requirement matrix $\D_f(t,\u_1,\dots,\u_{2^k})$ and the corresponding value $N(\D_f(t,f_1,\dots,f_\I))$, especially when $\I$ is small. 

\subsection{Irregular-Distance Codes}

We summarize some results about $N(\D)$ here, which allow us to obtain results on the redundancy of \fccs~using Theorems \ref{thm:optimal:fcc:sub} and \ref{thm:achievability}. We start with a generalization of the Plotkin bound \cite{plotkin_binary_1960} on codes with irregular distance requirements.
\begin{lemma} \label{lemma:irregular:distance:code:plotkin} For any distance matrix $\D \in \mathbb{N}_0^{M\times M}$, 	
	$$ N(\D) \geq \begin{cases}
		\frac{4}{M^2} \sum\limits_{i,j:i<j}[\D]_{ij},  \text{ if } M \text{ is even}, \\ \frac{4}{M^2-1} \sum\limits_{i,j:i<j}[\D]_{ij},  \text{ if } M \text{ is odd.}
	\end{cases} $$
\end{lemma}
\begin{proof}
	We start by proving the statement for $M$ even. Let $\p_1,\dots,\p_M$ be codewords of a $\D$-code of length $r$. Stack these codewords as rows of a matrix $\matP$. Since each column of the matrix $\matP$ can contribute at most $\frac{M^2}{4}$ to the sum $\sum_{i,j:i<j}d(\p_i,\p_j)$ (when the weight of the column is exactly $\frac{M}{2}$), we have that
	$ \sum_{i,j:i<j}d(\p_i,\p_j) \leq r \frac{M^2}{4}. $
	On the other hand, by the definition of a $\D$-code, $d(\p_i,\p_j) \geq [\D]_{ij}$  and the statement follows. The statement for odd $M$ is proven accordingly using the fact that in this case the maximum contribution of a column is $\frac{M+1}{2}\frac{M-1}{2}$. 
\end{proof}
For the case of regular-distance codes with minimum distance $D$, Lemma~\ref{lemma:irregular:distance:code:plotkin} implies $N(M,D) \geq 2D \frac{M-1}{M}$, a variant of Plotkin's bound. Conversely, we can derive an achievability bound, which is a generalization of the well-known Gilbert-Varshamov bound \cite{gilbert_comparison_1952,varshamov_estimate_1957} to irregular-distance codes. To this end we define $V(r,d) = \sum_{i=0}^{d}\binom{r}{i}$ as the size of the binary radius-$d$ Hamming sphere over vectors of length $r$. 
\begin{lemma} \label{lemma:irregular:distance:code:gv} For any distance matrix $\D \in \mathbb{N}_0^{M\times M}$, and any permutation $\pi: [M] \to [M]$
	$$ N(\D) \leq {\min_{r\in \mathbb{N}}}~\left\{r:2^{r} > \underset{j \in [M]}{\max} \sum_{i=1}^{j-1} V(r,[\D]_{\pi(i)\pi(j)}-1)\right\}. $$
\end{lemma}
\begin{proof}
	We describe how to construct a code of length $r$ meeting the distance requirements by iteratively selecting valid codewords. Assume first for simplicity that $\pi(i)=i$. Start by choosing an arbitrary codeword $\p_1 \in \Z_2^r$. Then, choose a valid codeword $\p_2$ as follows. Since the distance of $\p_1$ and $\p_2$ needs to be at least $[\D]_{12}$, we choose an arbitrary $\p_2$ such that $d(\p_1,\p_2) \geq [\D]_{12}$. Such a codeword $\p_2$ exists, if the length satisfies $2^{r} > V(r,[\D]_{12}-1)$. Next, we choose the third codeword $\p_3$. Similarly as before, we need to have $d(\p_1,\p_3) \geq [\D]_{13}$ and also $d(\p_2,\p_3) \geq [\D]_{23}$. If $2^{r} > V(r,[\D]_{13}-1) +V(r,[\D]_{23}-1)$ we can guarantee the existence of such a codeword $\p_3$. The theorem then follows by iteratively selecting the remaining codewords $\p_j$ such that $d(\p_i,\p_j) \geq [\D]_{ij}$ for all $i<j$. Under the condition of the theorem, we can guarantee existence of all codewords. Since the codewords can be chosen in an arbitrary order, the lemma holds for any order $\pi$ in which the codewords are selected.
\end{proof}
Note that for codes with $[\D]_{ij}=D$, this bound results in the well-known Gilbert-Varshamov bound \cite{gilbert_comparison_1952,varshamov_estimate_1957}.

Several of our results in the following require codes of small cardinality, i.e., the code size is in the same order of magnitude as the minimum distance. The following result is based on Hadamard codes \cite{macwilliams_theory_2007,horadam_hadamard_2007}.
\begin{lemma} (cf. \cite[Def. 3.13]{horadam_hadamard_2007}) \label{lemma:regular:distance:hadamard}
	Let $D \in \mathbb{N}$ be such that there exists a Hadamard matrix of order $D$ and $M \leq 4D$. Then,
	$$ N(M,D) \leq 2D. $$
\end{lemma}
The range of the parameter $D$ is restricted to the limited knowledge of lengths for which Hadamard codes exist. Note that there exist other good codes of small size, such as weak flip codes \cite{lin_weak_2018-1}, however, they only attain the Plotkin bound for a limited range of parameters. In general, it is possible to puncture or juxtapose Hadamard codes (cf. Levenshtein's theorem \cite[Section 2.3]{macwilliams_theory_2007}) to obtain codes for a larger range of parameters. However, for our discussion, the application of the Gilbert-Varshamov bound is sufficient and further allows to prove the existence of codes whose size is quadratic in their minimum distance as follows.
\begin{lemma} \label{lemma:regular:distance:code} For any $M,D \in \mathbb{N}$ with $D \geq 10$ and $M \leq D^2$,
	$$ 	N(M,D) \leq \frac{2D}{1-2\sqrt{\ln(D)/D}}. $$ 
\end{lemma}
The proof of Lemma \ref{lemma:regular:distance:code} is obtained using Lemma \ref{lemma:irregular:distance:code:gv} together with \cite[Lemma 4.7.2]{ash_information_1990} and is presented in Appendix~\ref{app:proof_lemma4}.
This result means that, given that the size of the code is moderate, i.e., $M\leq D^2$, for large $D$, the optimal length of an error-correcting code approaches $2D$. While Lemma \ref{lemma:regular:distance:code} gives a slightly weaker bound than Lemma \ref{lemma:regular:distance:hadamard}, it holds for any $D$ and for larger code sizes $M$. Note that a similar bound as in Lemma \ref{lemma:regular:distance:code} can easily be derived also for larger $M$, i.e., $M \leq D^m$, $m >2$, however $m=2$ is sufficient for the subsequent analysis. Denoting $D_{\max} = \max_{i,j} [\D]_{ij}$, with these bounds it is immediate that, if $M \leq D_{\max}^2$, it holds that $D_{\max} \leq N(\D) \leq 2D_{\max}/{(1-2\sqrt{\ln(D_{\max})/D_{\max}})}$.

In the following sections, we turn to discuss specific functions and give bounds on their optimal redundancy, which are tight in several cases. For several instances we additionally give explicit code constructions that can be encoded efficiently. The functions under discussion are locally binary functions, the Hamming weight function, the Hamming weight distribution function, the min-max function and a collection of discretized real-valued functions.

\section{Locally Binary Functions}\label{sec:locally_binary}
In the following we define a broad class of functions, called locally binary functions. We derive their optimal redundancy and show how it can be obtained using a simple explicit code construction. This class of functions is defined next. %
\begin{defn}
    The function ball of a function $f$ with radius $\rho$ around $\u\in\Z_2^k$ is defined by
    $$ B_f(\u,\rho) = \{f(\u'): \u' \in \Z_2^k \land d(\u,\u') \leq \rho\}.  $$
\end{defn}
Locally binary functions are defined as follows.
\begin{defn}
	A function $f:\Z_2^k\to \Im(f)$ is called a $\rho$-locally binary function, if for all $\u \in \Z_2^k$,
	$$ |B_f(\u,\rho)| \leq 2. $$
\end{defn}
Intuitively, a $\rho$-locally binary function is a function, where the function regions of all function values are well spread in the sense that each information word is close to only one region of another function value, see Fig.~\ref{fig:locally:binary}.
\begin{figure}
	\centering
	\begin{subfigure}[t]{0.3\textwidth}
		\begin{tikzpicture}[scale=1.35]
			\draw[opacity=0.5] (0,0) rectangle (4,4);
			\draw[pattern=north east lines,opacity=0.2] (0,0) -- (1,0) .. controls (0.5,0.7) .. (0,1) -- cycle;
			\draw[pattern=dots,opacity=0.2] (1,0) .. controls (0.5,0.7) .. (0,1) -- (0,3.1) .. controls (1.0,1.4) .. (2.5,0) -- cycle;
			\draw[pattern=north west lines,opacity=0.2] (0,3.1) .. controls (1.0,1.4) .. (2.5,0) -- (4,0) -- (4,4) -- (3,4) .. controls (2.8,2.7) and (1.6,2.6) .. (1.8,4) -- (0,4) -- (0,3.1);
			\draw[pattern=crosshatch dots,opacity=0.2] (3,4) .. controls (2.8,2.7) and (1.6,2.6) .. (1.8,4) -- cycle;
			\node at (0.3,0.3) {$f_1$};
			\node at (.8,1.2) {$f_2$};
			\node at (2,2) {$f_3$};
			\node at (2.3,3.5) {$f_4$};
			\node at (0.3,3.7) {$\Z_2^k$};
		\end{tikzpicture}
	\caption{Example of a locally-binary function for small $\rho$. The regions are well-separated and in each neighborhood, there exist only two distinct function values.}
	\end{subfigure}%
\hspace{1.2cm}
	\begin{subfigure}[t]{0.3\textwidth}
		\begin{tikzpicture}[scale=1.35]
			\draw[opacity=0.5] (0,0) rectangle (4,4);
			\draw[pattern=north east lines,opacity=0.2] (0,0) .. controls (1,1.5) .. (2,2) .. controls (2.2,1.2) .. (4,0) -- cycle;
			\draw[pattern=north west lines,opacity=0.2] (0,0) .. controls (1,1.5) .. (2,2) .. controls (2.1,3.2) .. (2,4) -- (0,4) -- cycle;
			\filldraw[pattern=dots,opacity=0.2] (2,2) .. controls (2.1,3.2) .. (2,4) -- (4,4) -- (4,0 ) .. controls (2.2,1.2) .. (2,2);
			\node at (2,0.5) {$f_1$};
			\node at (1,2.5) {$f_2$};
			\node at (3,2.5) {$f_3$};
			\node at (0.3,3.7) {$\Z_2^k$};
		\end{tikzpicture}
		\caption{Example of a \emph{non-}locally-binary function $f$. Information vectors $\u$ close to the intersection point in the middle have $|B_f(\u,\rho)|=3$. }
	\end{subfigure}%
	\caption{Visualization of (non-)locally binary functions. Shaded areas highlight function regions, i.e., $\{\u \in \Z_2^k : f(\u) = f_i \}$. }
	\label{fig:locally:binary}
\end{figure}
Note that by this definition, any binary function, i.e., $|\Im(f)|=2$, is also a $\rho$-locally binary function for arbitrary $\rho$. We can directly prove the following optimality.
\begin{lemma} \label{lemma:locally:binary}
	For any $2t$-locally binary function $f$,
	$$ r_f(k,t) = 2t. $$
\end{lemma}
\begin{proof}
	By Corollary \ref{cor:lower:bound}, $r_f(k,t) \geq 2t$. On the other hand, we can prove achievability using the following explicit code construction. Let $\Im(f) = \{f_1,\dots,f_\I\}$ and set w.l.o.g. $f_i\triangleq i$. Let $\u$ be the information word to be encoded and define the following function,
	$$ \omega_{2t}(\u) = \left\{ \begin{array}{ll}
		1, & \text{if } ~ f(\u) = \max B_f(\u,2t), \\
		0, & \text{otherwise.}
	\end{array} \right. $$

	Now, use $\Enc(\u) = (\u,(\omega_{2t}(\u))^{2t})$, i.e. the $2t$-fold repetition of the bit $\omega_{2t}(\u)$. This gives an \fcc~for the function $f$ due to the following. Assume $(\u,\p) =\Enc(\u)$ has been transmitted and $(\u',\p')$ has been received. The decoder first computes $B_f(\u',t)$. Notice that $f(\u) \in B_f(\u',t) \subseteq B_f(\u,2t)$. If $|B_f(\u',t)|=1$, then it trivially contains the correct function value $f(\u)$. Otherwise, $B_f(\u',t) = B_f(\u,2t)$, since $|B_f(\u',t)|>1$ and, by the definition of $2t$-locally binary functions, $|B_f(\u,2t)|\leq2$. The decoder performs a majority decision over the $2t+1$ bits $(\omega_{2t}(\u'),\p')$ and obtains correctly $\omega_{2t}(\u)$, as at most $t$ out of these $2t+1$ bits are erroneous. Finally, the receiver decides for $\max B_f(\u',t)$, if $\omega_{2t}(\u) = 1$ and for $\min B_f(\u',t)$, otherwise.
\end{proof}
It is noteworthy that the code construction used in Lemma~\ref{lemma:locally:binary} leverages the side information provided by the message $\u$ using $\omega_{2t}(\u')$ for decoding, which allows to achieve a redundancy of only $2t$. This side information is particularly useful for locally binary functions due to the structured topology of the function regions, which is visualized in Fig.~\ref{fig:locally:binary}. Ignoring this side information would require significantly more redundancy, cf. Table~\ref{tab:redundancy}.

In Section \ref{sec:hamming_d} we will present an explicit example of a locally binary function. For illustration, another example of a locally binary function is presented in the following. 
\begin{example}
Assume the codewords $\mathcal{Q} = \{\mathbf{q}_1,\dots,\mathbf{q}_N\} \subseteq \Z_2^k$ form a code of length $k$ with minimum distance $d=\min_{i\neq j} d(\mathbf{q}_i,\mathbf{q}_j)$. Then the indicator function 
$$\mathbb{I}_\mathcal{Q}(\u) = \begin{cases}
	i, & \text{if } \u=\mathbf{q}_i, \\
	0, & \text{otherwise}
\end{cases} $$
is $\lfloor\frac{d-1}{2}\rfloor$-locally binary.
\end{example}

\section{Functions Based on the Hamming Weight}\label{sec:hamming}
In this section we study two functions: the Hamming weight function $f(\u) = \wt(\u)$ and the Hamming weight distribution function $f(\u) = \Delta_T(\u) = \lfloor \frac{\wt(\u)}{T} \rfloor$, for a given threshold $T$.
\subsection{Hamming Weight Function}\label{sec:hamming_w}

Let $f(\u) = \wt(\u)$, where $\u \in \Z_2^k$. Note that the expressiveness of $\wt(\bullet)$ is $\I=|\Im(\wt)|=k+1$. We start by showing that for this function it is possible to achieve optimal redundancy by an encoding function which only depends on the function value, i.e., the Hamming weight of $\u$. Throughout this section we refer to the function distance matrix $\D_{\mathrm{wt}}(t,f_1,\dots,f_\I)$ as $\D_{\mathrm{wt}}(t)$ for ease of notation. %
\begin{lemma} \label{lemma:optimal:hamming:weight:function}
	Let $f(\u) = \wt(\u)$. Consider the $(k+1)\times(k+1)$ matrix $\D_{\wt}(t)$ with entries $[\D_{\wt}(t)]_{ii} = 0$ and $[\D_{\wt}(t)]_{ij} = [2t+1-|i-j|]^+$ for $i \neq j$. Then,
	$$ r_{\wt}(k,t) = N(\D_{\wt}(t)). $$
\end{lemma}
\begin{proof}
	The function values of the Hamming weight function belong to $\Im(\wt) = \{0,1,\dots,k\}$ and we let $i,j \in \{0,1,\dots,k\}$ denote two function values. First, we see that the function distance is given by $d_{\wt}(i,j) = |i-j|$, since
	\begin{align*}
	    \min_{\u_1,\u_2 \in \Z_2^k}  d(\u_1,\u_2) ~\mathrm{s.t.}~ \wt(\u_1) = i, \wt(\u_2) = j
	\end{align*}
	is equal to $|i-j|$. It follows from Theorem \ref{thm:achievability} that $r_{\wt}(k,t) \leq N(\D_{\wt}(t))$. On the other hand, using $\u_i = (1^{i}0^{k-i})$, \mbox{$i \in \{0,1,\dots,k\}$}, we see that $\wt(\u_i) = i$ and their pairwise distances are $d(\u_i,\u_j) = |i-j|$. We can then apply Corollary~\ref{cor:lower:bound} to obtain $r_{\wt}(k,t) \geq N(\D_{\wt}(t))$.
\end{proof}
The following example visualizes the general structure of the function distance matrix $\D_{\wt}(t)$.
\begin{example}
	The function distance matrix $\D_{\wt}(2)$ for $k=6$ is given by the symmetric $7 \times 7$ matrix with entries
	$$ \D_{\wt}(2) = \begin{blockarray}{c|@{\hspace{8pt}}ccccccc}
		$f(\u)$ & 0 & 1 & 2 & 3 & 4 & 5 & 6 \\ \cline{1-8}
		\begin{block}{c|@{\hspace{8pt}}(ccccccc)}
			0&\cellcolor{gray!15}0&4&3&2&1&0&0\\
			1&\textcolor{gray}{4}&\cellcolor{gray!15}0&4&3&2&1&0\\
			2&\textcolor{gray}{3}&\textcolor{gray}{4}&\cellcolor{gray!15}0&4&3&2&1\\
			3&\textcolor{gray}{2}&\textcolor{gray}{3}&\textcolor{gray}{4}&\cellcolor{gray!15}0&4&3&2\\
			4&\textcolor{gray}{1}&\textcolor{gray}{2}&\textcolor{gray}{3}&\textcolor{gray}{4}&\cellcolor{gray!15}0&4&3\\
			5&\textcolor{gray}{0}&\textcolor{gray}{1}&\textcolor{gray}{2}&\textcolor{gray}{3}&\textcolor{gray}{4}&\cellcolor{gray!15}0&4\\
			6&\textcolor{gray}{0}&\textcolor{gray}{0}&\textcolor{gray}{1}&\textcolor{gray}{2}&\textcolor{gray}{3}&\textcolor{gray}{4}&\cellcolor{gray!15}0\\
		\end{block}
	\end{blockarray} .$$
\end{example}
Based on Lemma \ref{lemma:optimal:hamming:weight:function}, we can infer a lower bound on the redundancy using the Plotkin-like bound of Lemma \ref{lemma:irregular:distance:code:plotkin}. %
\begin{corollary} \label{cor:hamming:weight:lower:bound} For any $k>t$,
	$$ r_{\wt}(k,t) \geq \frac{10t^3+30t^2+20t+12}{3t^2+12t+12}. $$
\end{corollary}
\begin{proof}
	Let $\{\p_1,\dots,\p_{k+1}\}$ be a $\D_{\wt}(t)$-code. We will prove the corollary by applying the Plotkin-type bound on a subcode of $\p_1,\dots,\p_{k+1}$. Consider the first $t+2$ codewords $\p_1,\dots,\p_{t+2}$. By Lemma \ref{lemma:optimal:hamming:weight:function}, we have that $[\D_{\wt}(t)]_{ij} = 2t+1-|i-j|$ and thus $[\D_{\wt}(t)]_{12}+[\D_{\wt}(t)]_{13}+[\D_{\wt}(t)]_{23} = 6t-1$. However, since $d(\p_1,\p_2)+d(\p_1,\p_3)+d(\p_2,\p_3)$ must be an even value, it follows that $d(\p_1,\p_2)+d(\p_1,\p_3)+d(\p_2,\p_3) \geq 6t$. With this strengthened bound, the sum of the pairwise distances in Lemma \ref{lemma:irregular:distance:code:plotkin} can be increased by one and we obtain
	\begin{align*}
		r_{\wt}(k,t) &\overset{(a)}{\geq} \frac{4}{(t+2)^2} \left(1+ \sum_{i=1}^{t+2}\sum_{j=i+1}^{t+2} [\D_{\wt}(t)]_{ij} \right) \\
		&\overset{(b)}{=} \frac{4}{(t+2)^2} \left(1+\sum_{i=0}^{t} (t+1-i)(2t-i)\right) \\
		&= \frac{10t^3+30t^2+20t+12}{3t^2+12t+12}.
	\end{align*}
	Hereby, inequality $(a)$ follows from Lemma \ref{lemma:irregular:distance:code:plotkin}, with an additional summand of $1$ due to the fact that $d(\p_1,\p_2)+d(\p_1,\p_3)+d(\p_2,\p_3)$ must be even, as explained above. Eq. $(b)$ follows from summing over the diagonals of $\D_{\wt}(t)$.
\end{proof}
For the following results, we require the \emph{shifted} modulo function, which is defined as follows.
\begin{defn} \label{def:smod}
	We define the shifted modulo operator by
	$$a \smod b \triangleq ((a-1) \bmod b) +1 \in \{1,2,\dots,b\}.$$
\end{defn}
E.g., \mbox{$(3 \smod 3) = 3$} and $(4 \smod 3) = 1$. We now describe a construction of an FCC for the function $\wt(\u)$.
\begin{construction}
	We define
	$$\Enc_{\wt}(\u) = (\u,\p_{\wt(\u)+1}),$$
	where the $\p_i$'s are defined depending on $t$ as follows.
	
	For $t=1$, set $\p_1 = (000)$, $\p_2 = (110)$ and $\p_3 = (011)$. Then set $\p_i = \p_{i\smod3}$ for $i\geq4$.
	
	For $t=2$ set $\p_1 = (000000)$, $\p_2 = (110011)$, $\p_3 = (001111)$, $\p_4 = (111100)$. Then set $\p_i=\p_{i-4}+(000001)$ for $i \in \{5,6,7,8\}$ and $\p_i=\p_{i\smod8}$ for $i\geq 9$.
	
	For $t \geq 3$, let $\p_1,\dots,\p_{2t+1}$ be a code with minimum distance $2t$, i.e., $d(\p_i,\p_j) \geq 2t$ for all $i,j\leq 2t+1$, $i\neq j$ and set $\p_i = \p_{i\smod (2t+1)}$ for $i \geq 2t+2$.

\end{construction}
We can use Corollary \ref{cor:hamming:weight:lower:bound} to narrow down the optimal redundancy of \fccs~for the Hamming weight function as follows.
\begin{lemma} \label{lemma:hamming:weight:fct}
	For any $k>2$, $r_{\wt}(k,1) = 3$ and $r_{\wt}(k,2) = 6$. Further, for $t \geq 5$ and $k>t$,
	$$ \frac{10t}{3}-\frac{10}{3} \leq r_{\wt}(k,t) \leq \frac{4t}{1-2\sqrt{\ln(2t)/(2t)}}. $$
\end{lemma}
\begin{proof}
	We start with the case $t=1$. It is quickly verified that $d(\p_i,\p_j) \geq [\D_{\wt}(1)]_{ij}$ for all $i\neq j$, $i,j\leq k+1$ and thus giving a valid \fcc. Further, Corollary \ref{cor:hamming:weight:lower:bound} gives $r_{\wt}(k,1) \geq 3$.
	For the case $t=2$, it can be verified that $d(\p_i,\p_j) \geq [\D_{\wt}(2)]_{ij}$. Again, Corollary \ref{cor:hamming:weight:lower:bound} gives $r_{\wt}(k,1) \geq 6$, proving optimality of the proposed code.
	For $t \geq 3$, we obtain $d(\p_i,\p_j) \geq [\D_{\wt}(t)]_{ij}$ as desired. The lower and upper bound on $r_{\wt}(k,t)$ follow from Corollary \ref{cor:hamming:weight:lower:bound} and Lemma \ref{lemma:regular:distance:code}.
\end{proof}
Recall here that using a standard error-correcting code with minimum distance $2t+1$, e.g., a BCH code, results in a redundancy of roughly $t \log k$. Therefore, using \fccs, we can improve the scaling of the redundancy by a factor of $\log k$. While we find the optimal redundancy exactly for $t=1$ and $t=2$, there is still a gap for $t \geq 3$ narrowing down the optimal redundancy between roughly $\frac{10t}{3}$ and $4t$.

\subsection{Hamming Weight Distribution Function}\label{sec:hamming_d}
Let in the following $T \in \mathbb{N}$ be a parameter of choice. For simplicity, we restrict $T$ to divide $k+1$. Consider the function $f(\u) = \Delta_T(\u) \triangleq \lfloor \frac{\wt(\u)}{T} \rfloor$. We directly see that the number of distinct function values is equal to $\I = \frac{k+1}{T}$. This function defines a step threshold function, based on the Hamming weight of $\u$, with $\I-1$ steps. The threshold values, where the function values increase by one, are at integer multiples of $T$, see Fig. \ref{fig:wt:distr}.
\pgfplotsset{compat=newest}
\begin{figure}
	\centering
	\begin{tikzpicture}
		\begin{axis}[set layers,
			xlabel={$\wt(\u)$},
			ylabel={$\Delta_T(\u)$},
			xmin=0,xmax=6,
			ymin=0,ymax=5.5,
			height=6cm,
			xtick={0,1,2,3,4,5,6},
			ytick={0,1,2,3,4,5},
			xticklabels={0,$T$, $2T$, $\dots$ ,, , $k$},
			yticklabels={0,1,2,$\vdots$, $E-2$, $E-1$},
			set layers,
			minor x tick num=4,
			xminorticks=true,
			x tick label style={yshift={(\ticknum==3)*-0.5em},xshift={(\ticknum==6)*-0.5em}},]
			\node[] at (axis cs: 3.25,2.8) {$\iddots$};

			\addplot[mark=*, only marks] coordinates {(0,0) (0.2,0) (0.4,0) (0.6,0) (0.8,0) (1,1) (1.2,1) (1.4,1) (1.6,1) (1.8,1) (2,2) (2.2,2) (2.4,2) (2.6,2) (2.8,2) (4,4) (4.2,4) (4.4,4) (4.6,4) (4.8,4) (5,5) (5.2,5) (5.4,5) (5.6,5) (5.8,5)};

			\pgfplotsinvokeforeach {1,2,3,4,5}{
				\pgfonlayer{axis grid}
				\draw[gray!50](#1,0|-current axis.south)--(#1,0|-current axis.north);
				\endpgfonlayer}
			\pgfplotsinvokeforeach {1,2,3,4,5}{
				\pgfonlayer{axis grid}
				\draw[gray!50](0,#1-|current axis.west)--(0,#1-|current axis.east);
				\endpgfonlayer
			}
		\end{axis}
		
	\end{tikzpicture}
	\caption{Illustration of the Hamming weight distribution function with step size $T = \frac{k+1}{\I}$.}
	\label{fig:wt:distr}
\end{figure}
We restrict to the case where $2t+1 \leq T$ and will give an optimal construction with redundancy $r_{\Delta_T}(k,t) = 2t$ in this regime. First, note that, when $4t+1 \leq T$, we can show that $\Delta_T(\u)$ is $2t$-locally binary, as two consecutive thresholds have distance at least $4t+1$. Consequently, $r_{\Delta_T}(k,t) = 2t$ by Lemma \ref{lemma:locally:binary}. We now focus on the more general case, where $2t+1 \leq T$. We start by describing the encoding function. Recall the shifted modulo operation from Definition~\ref{def:smod}.
\begin{construction}\label{cons:fcc:Hwd} We define
	$$\Enc_{\Delta_T}(\u) = (\u,\p_{\wt(\u)}),$$
	with $\p_i \in \Z_2^{2t}$ defined as follows. Set $\p_i = (1^{i-1}0^{2t-i+1})$ for $i \in [2t+1]$, $\p_i = (1^{2t})$ for $i \in \{2t+2,\dots,T\}$ and $\p_i = \p_{i\smod T}$, if $i \geq T+1$.
\end{construction}
We show that this encoding function gives an {\fcc} for the Hamming weight distribution function $\Delta_T(\u)$.
\begin{lemma}\label{lemma:optimal:fcc:Hwd}
	For any $k,t,T\in \mathbb{N}$ such that $T$ divides $(k+1)$ and $2t+1\leq T$, 
	$$r_{\Delta_T}(k,t) = 2t.$$
\end{lemma}
\begin{proof}
	By Corollary~\ref{cor:lower:bound}, $r_{\Delta_T}(k,t) \geq 2t$. We now argue that Construction \ref{cons:fcc:Hwd} is an {\fcc} of redundancy $2t$ by showing that $d(\Enc_{\Delta_T}(\u_1),\Enc_{\Delta_T}(\u_2))\geq 2t+1$ for all $\u_1,\u_2 \in \Z_2^k$ with $f(\u_1) \neq f(\u_2)$. Let $\u_1,\u_2 \in \Z_2^k$ with $f(\u_1) \neq f(\u_2)$ be two information vectors that evaluate to two different function values. Note that, if $d(\u_1,\u_2) \geq 2t+1$, we automatically have $d(\Enc_{\Delta_T}(\u_1),\Enc_{\Delta_T}(\u_2))\geq 2t+1$ and we therefore restrict to the interesting case $d(\u_1,\u_2) <2t+1$. Since $f(\u_1) \neq f(\u_2)$ and $T\geq 2t+1$ we can therefore assume w.l.o.g. that $f(\u_1) = m-1$ and $f(\u_2)=m$ for some $m \in \mathbb{N}$.
	
	We will prove the lemma for $T=2t+1$ first. In this case, the parity vectors $\p_i$ in the two function regions are illustrated in Table~\ref{tab:parity:vectors:threshold}.
	\begin{table}
		\caption{Parity vectors of the Hamming weight distribution function for $T=2t+1$.}
		\begin{center}\setlength{\tabcolsep}{12.5pt}\renewcommand{\arraystretch}{1.15}
			\begin{tabular}{ccc}
				$\wt(\u)$ & $\p_{\wt(\u)+1}$ & $f(\u)$ \\ \specialrule{.4pt}{0pt}{0pt}
				$mT-T$ & $(000\dots000)$ & \multirow{5}{*}{$m-1$} \\[-2pt]
				$mT-T+1$ & $(100\dots000)$ & \\[-2pt]
				$\vdots$& $\vdots$ & \\[1pt]
				$mT-2$ & $(111\dots110)$ &  \\
				$mT-1$ & $(111\dots111)$ & \\ \specialrule{.4pt}{0pt}{0pt}
				$mT$ & $(000\dots000)$ & \multirow{3}{*}{$m$} \\
				$mT+1$ & $(100\dots000)$ & \\[-2pt]
				$\vdots$& $\vdots$ &
			\end{tabular}
		\end{center}
	\label{tab:parity:vectors:threshold}
	\end{table}
	Let $\wt(\u_1) = (m-1)T+w_1$ and $\wt(\u_2) = mT+w_2$ with $w_1,w_2 \in \{0,1,\dots,T-1\}$. The corresponding parity vectors are $\p(\u_1) = \p_{w_1+1} = (1^{w_1}0^{2t-w_1})$ and $\p(\u_2) = \p_{w_2+1} = (1^{w_2}0^{2t-w_2})$. Using $d(\p_{w_1+1},\p_{w_2+1}) = |w_1-w_2|$, it follows that $d(\Enc_{\Delta_T}(\u_1),\Enc_{\Delta_T}(\u_2)) = d(\u_1,\u_2) + d(\p_{w_1+1},\p_{w_2+1}) \geq \wt(\u_2)-\wt(\u_1) + |w_1-w_2| = T-w_1+w_2+|w_1-w_2|\geq T=2t+1$. The case $T>2t+1$ is proven similarly using that
	$$ d(\p_{w_1+1},\p_{w_2+1}) = \begin{cases}
		|w_1-w_2|, & \text{ if } w_1 \leq 2t \land w_2 \leq 2t, \\
		[2t-w_2]^+, & \text{ if } w_1 > 2t,  \\
		[2t-w_1]^+, & \text{ if } w_2 > 2t.  \\
	\end{cases} $$
\end{proof}

\section{Min-Max Functions}\label{sec:min_max}
\newcommand{\uii}[1]{\u^{(#1)}}
\newcommand{\ui}{\u^{(i)}}
\newcommand{\uj}{\u^{(j)}}
\newcommand{\uv}{\u^{(v)}}
Assume now that $k = w\ell$ for some integers $w$ and $\ell$. In this section, we consider $\u$ to be formed of $w$ parts, such that $\u=(\u^{(1)},\dots,\u^{(w)})$, where each $\u^{(i)} \in \Z_2^\ell$ is of length $\ell$. The function of interest is the min-max function defined next.

\begin{defn}\label{def:minmax} 
The min-max function is defined by
$$\minmax_w(\u) = (\underset{1\leq i\leq w}{\argmin}~ \u^{(i)},\underset{1\leq i\leq w}{\argmax}~ \u^{(i)}),$$
where  $\u=(\u^{(1)},\dots,\u^{(w)})$, $\u^{(i)} \in \Z_2^\ell$ with $k=w\ell$ and the ordering $<$ between the $\u^{(i)}$'s is primarily lexicographical (the left-most bit is the most significant) and secondarily, if $\ui=\uj$, according to ascending indices. %
\end{defn}

For example, $\u = (\uii{1},\uii{2},\uii{3}) = (100, 010,010)$, has the ordering $\uii{2}<\uii{3}<\uii{1}$ and thus $\minmax_w(\u) = (2,1)$.
For $w=1$, the function is constant and for $w=2$, the function is a binary function and we have an optimal solution from Lemma~\ref{lemma:locally:binary}. For $w\geq 3$, we provide two lower bounds on the redundancy in Lemma \ref{lemma:redandancy_lower_bound_mm} and Corollary \ref{cor:minmax:sp}. We characterize the function distance matrix of the min-max function in Claims \ref{claim:distance_ub} and \ref{claim:distance_profile} and obtain an upper bound on the redundancy based on Theorem~\ref{thm:achievability}, which is derived in Lemma \ref{lemma:redandancy_upper_bound_mm}. Since Lemma \ref{lemma:redandancy_upper_bound_mm} is obtained using a Gilbert-Varshmov argument, the result is of existential nature. We construct explicit \fccs\ based on standard error-correcting codes in Construction \ref{cons:minmax} {and Construction~\ref{cons:RMminmax}}. %
Throughout this section we refer to the function distance matrix $\D_{\mathrm{mm}}(t,f_1,\dots,f_\I)$ as $\D_{\mathrm{mm}}$ for ease of notation. 
The following example illustrates our results.

\begin{example}
Consider a min-max function with $w=3$ and $\ell \geq 3$. From Claim~\ref{claim:distance_ub} and Claim~\ref{claim:distance_profile} we obtain the function distance matrix $\D_{\mathrm{mm}}$ for this case and any $t$ as follows

\newlength{\colsep}
\setlength{\colsep}{4pt}
$$\D_{\mathrm{mm}} \!=\! \begin{blockarray}{c|@{\hspace{8pt}}c@{\hspace{\colsep}}c@{\hspace{\colsep}}c@{\hspace{\colsep}}c@{\hspace{\colsep}}c@{\hspace{\colsep}}c} 
	$f(\u)$ & (1,2) & (1,3) & (2,1) & (2,3) & (3,1) & (3,2) \\ \cline{1-7}
	\begin{block}{c|@{\hspace{8pt}}(c@{\hspace{\colsep}}c@{\hspace{\colsep}}c@{\hspace{\colsep}}c@{\hspace{\colsep}}c@{\hspace{\colsep}}c)}
		(1,2)&\cellcolor{gray!15}0&2t&2t-1&2t&2t&2t\\
		(1,3)&\textcolor{gray}{2t}&\cellcolor{gray!15}0&2t&2t&2t-1&2t\\
		(2,1)&\textcolor{gray}{2t-1}&\textcolor{gray}{2t}&\cellcolor{gray!15}0&2t&2t&2t\\
		(2,3)&\textcolor{gray}{2t}&\textcolor{gray}{2t}&\textcolor{gray}{2t}&\cellcolor{gray!15}0&2t&2t-1\\
		(3,1)&\textcolor{gray}{2t}&\textcolor{gray}{2t-1}&\textcolor{gray}{2t}&\textcolor{gray}{2t}&\cellcolor{gray!15}0&2t\\
		(3,2)&\textcolor{gray}{2t}&\textcolor{gray}{2t}&\textcolor{gray}{2t}&\textcolor{gray}{2t-1}&\textcolor{gray}{2t}&\cellcolor{gray!15}0\\
	\end{block}
\end{blockarray} .$$

For example, the function distance between the function values $(1,2)$ and $(1,3)$ is $1$ since there exist information words $\u_1 = (000,010,001)$ and $\u_2 = (000,010,0\mathbf{1}1)$ such that $d(\u_1,\u_2) = 1$ and $\minmax_w(\u_1) = (1,2)$ and $\minmax_w(\u_2) = (1,3)$. For $w=3$ this holds for all pairs of function values except for those of the form $(i,j)$ and $(j,i)$ where at least two bits must be changed to move from one function value to another, i.e., for every $\u_1,\u_2$ such that $\minmax_w(\u_1) = (i,j)$ and $\minmax_w(\u_2)=(j,i)$, we have $d(\u_1,\u_2)\geq 2$, cf. proof of Claim~\ref{claim:distance_profile}. A possible construction for an \fcc~is to use a code with cardinality $w(w-1)=6$ and distance matrix $\D_{\minmax}$ in the fashion of Theorem \ref{thm:achievability}, i.e., the redundancy vectors are encoded based on $f(\u)$ instead of $\u$. Note that we will observe later that such an encoding yields a redundancy that is not too far from optimality. From Lemma \ref{lemma:redandancy_lower_bound_mm}, which is presented in the sequel, for $w=3$ the optimal \fcc~redundancy is at least $10t/3-11/6$. On the other hand, using single-parity check codes, we will construct next an \fcc\ for $w=3$ with redundancy $r_\mathrm{SP} = 4t$ in Construction~\ref{cons:minmax}.

\end{example}

We now formally present our results. We start with the lower bound on the redundancy.

\begin{lemma}\label{lemma:redandancy_lower_bound_mm}
	For $w\geq 3$ and $\ell\geq 2$, the optimal redundancy $r_{\minmax_w}(k,t)$ is bounded from below by
	\begin{equation*}
	    r_{\minmax_w}(k,t)\geq \dfrac{4 t (w^2 - w - 1) - 3w^2 + 7w - 5}{(w - 1) w}.
	\end{equation*}
\end{lemma}

\begin{proof}%
	Let $\u_{i,j} \in \Z_2^k$ with $i,j \in [w]$, $i\neq j$ be $w(w-1)$ information vectors that will be specified later and $\D_{\mathrm{mm}}(t,\u_{1,2}, \dots,\u_{w-1,w})$ be their distance matrix. We use the result of Corollary~\ref{cor:lower:bound} and Lemma~\ref{lemma:irregular:distance:code:plotkin} to obtain
	\begin{align*}
	    r_{\minmax_w}\!(k,t)  &\geq N(\D_{\mathrm{mm}}(t,\u_{1,2}, \dots,\u_{w-1,w}))\\
	      & \geq \dfrac{4}{(w(w\!-\!1))^2}\!\!\sum_{i,j:i<j}\![\D_{\mathrm{mm}}(t,\u_{1,2}, \dots,\u_{w-1,w})]_{ij}.
	\end{align*}
	
	We first prove the lower bound for $\ell = 2$. To obtain a good lower bound, we need to find a suitable set of $w(w-1)$ representative information vectors and characterize their distance matrix $\D_{\mathrm{mm}}(t,\u_{1,2}, \dots,\u_{w-1,w})$. We choose the representative information vectors $\u_{i,j}$ to be 
	$$ \u_{i,j} =  (01,\dots,01,\underbrace{00}_{\u^{(i)}},\underbrace{11}_{\u^{(j)}}, 01,\dots,01), $$
	where $i,j \in [w]$ and $i\neq j$. Note that $\minmax_w(\u_{i,j}) = (i,j)$ and therefore the corresponding function values are all distinct. Let further $i',j' \in [w]$ with $i \neq i'$ and $j \neq j'$. We directly see that $d(\u_{i,j}, \u_{i,j'}) = d(\u_{i,j}, \u_{i',j}) = 2$ for function values which agree either in the minimum or maximum value. Further, $d(\u_{i,j}, \u_{i',j'}) = 4$ for function values that agree neither on the minimum nor maximum. For a given $\u_{i,j}$, there are thus
	\begin{itemize}
	  \item $(w-2)$ information vectors $\u_{i,j'}$ at distance $2$,
	  \item $(w-2)$ information vectors $\u_{i',j}$ at distance $2$,
	  \item $(w\!-\!1)(w\!-\!2)\!+\!1$ information vectors $\u_{i',j'}$ at distance $4$.
	\end{itemize}
	Therefore, each row of $\D_{\mathrm{mm}}(t,\u_{1,2}, \dots,\u_{w-1,w})$ has $2(w-2)$ entries that are equal to $2t-1$ and $(w-1)(w-2)+1$ entries that are equal to $2t-3$. Having characterized the values of the entries of the distance matrix, we can now write
	\begin{align*}
	    r_{\minmax_w}& \!(k,t) \geq \dfrac{4}{(w(w\!-\!1))^2}\sum_{i,j:i<j}\![\D_{\mathrm{mm}}(t,\u_{1,2}, \dots,\u_{w-1,w})]_{ij} \label{eq:lb1}\\
	      & \overset{(a)}{=} \dfrac{2}{(w(w-1))^2}\sum_{i,j}[\D_{\mathrm{mm}}(t,\u_{1,2}, \dots,\u_{w-1,w})]_{ij} \\
	      & \overset{(b)}{=} \dfrac{4 t (w^2 - w - 1) - 3w^2 + 7w - 5}{(w - 1) w}.
	\end{align*}
	Equation $(a)$ follows from the symmetry of the matrix $\D_{\mathrm{mm}}(t,\u_{1,2}, \dots,\u_{w-1,w})$ and equality $(b)$ follows by replacing the values discussed above and rearranging the terms. This proves the lower bound of Lemma \ref{lemma:redandancy_lower_bound_mm}. The proof for all $\ell>2$ follows the same steps after setting the $\ell-2$ left-most bits in every part of each $\u_{i,j}$ to $0$. 
\end{proof}

While this bound provides a good bound for large $t$ and moderate $w$, we can derive a stronger bound for fixed $t$ and large $w$ as follows.
\begin{corollary} \label{cor:minmax:sp}
	For $w \geq 3$ and $\ell \geq 2$, the optimal redundancy $r_{\minmax_w}(k,t)$ is bounded from below by
	$$ r_{\minmax_w}(k,t) \geq \log (w(w-1)) + (t-2) \log \log (w(w-1)) -t\log t. $$
\end{corollary}
\begin{proof}
	From the proof of Lemma \ref{lemma:redandancy_lower_bound_mm}, we know that $r_{\minmax_w}(k,t) \geq N(\D_{\minmax}(t,\u_{1,2}, \dots,\u_{w-1,w}))$. This quantity however can be bounded from below by noting that $d(\u_{i,j},\u_{i',j'}) \leq 4$ for any $i,j,i',j'$ (as shown in the same proof) and thus $N(\D_{\minmax}(t,\u_{1,2}, \dots,\u_{w-1,w})) \geq N(w(w-1), 2t-3)$. In other words, the $w(w-1)$ vectors must form a code of minimum distance $2t-3$. Abbreviating $r \triangleq N(w(w-1),2t-3)$ it follows from a sphere packing argument that $2^r \geq w(w-1) V(r,t-2)$, where $V(r,t)=\sum_{i=0}^{t}\binom{r}{i}$ is the size of the radius-$t$ Hamming sphere over vectors of length $r$. Consequently,
	\begin{align*}
		r &\geq \log w(w-1) + \log V(r,t-2) \\
			&\geq \log w(w-1) + (t-2) \log (r/(t-2)) \\
			& \overset{(a)}{\geq} \log w(w-1) + (t-2)\log \log w(w-1)-(t-2) \log t,
	\end{align*}
	where in $(a)$, we used the inequality $r \geq \log w(w-1)$.
\end{proof}
We provide two upper bounds on the optimal redundancy $r_{\minmax_w}(k,t)$ of \fccs\ designed for the min-max function. The first bound (Corollary~\ref{cor:mm:regular:ecc}) follows from Lemma~\ref{lemma:regular:distance:code} and uses standard error-correcting codes. On the other hand, the second bound (Lemma~\ref{lemma:redandancy_upper_bound_mm}) is obtained by examining the function distance matrix of the min-max function and using irregular-distance error-correcting codes.
\begin{corollary}[Corollary of Lemma~\ref{lemma:regular:distance:code}]\label{cor:mm:regular:ecc}
Given $t$ and $w$ such that $t\geq 5$ and $w(w-1)\leq 4t^2$, the optimal redundancy $r_{\minmax_w}(k,t)$ is bounded from above by
	$$ r_{\minmax_w}(k,t) \leq \frac{4t}{1-2\sqrt{\ln(2t)/2t}}. $$ 
\end{corollary}

\begin{proof}
Encoding the parity vectors with an error-correcting code of minimum distance $2t$ results in an \fcc . The redundancy of this \fcc\ is then equal to the length of the used code. Therefore, the bound holds from Lemma~\ref{lemma:regular:distance:code}.
\end{proof}

\begin{lemma}\label{lemma:redandancy_upper_bound_mm}
	For $w\geq 3$ and $\ell\geq 3$, the optimal redundancy $r_{\minmax_w}(k,t)$ of \fccs\ is bounded from above by
    \begin{align*}
        r_{\minmax_w}(k,t) %
        & \leq {\min_{r\in \mathbb{N}}}~\left\{r:\Phi(r)>0\right\},
    \end{align*}
    where 
    $\Phi(r) \triangleq 2^{r} - (w^2-w-1) V(r,2t-2)+ (4w-8)\binom{r}{2t-1}.$
\end{lemma}

\begin{proof}%

We start by bounding the distance between any two function values.

\begin{claim}\label{claim:distance_ub}
Consider a min-max function as defined in Definition~\ref{def:minmax}. For all $w\geq 3$ and $\ell \geq 3$ the minimum distance between any two function values (cf. Definition~\ref{def:function_distance}) $f_1$ and $f_2$ is at most $2$, i.e.,%
\begin{equation*}
\forall f_1,f_2 \in \mathrm{Im}(\minmax_w),\quad  d_{\minmax_w}(f_1,f_2) \leq 2.
\end{equation*}
\end{claim}

To prove Claim~\ref{claim:distance_ub} we need to show that for every two function values $f_1\neq f_2$, there exist two information vectors $\u_1\neq \u_2$ such that $\minmax_w(\u_1)=f_1$, $\minmax_w(\u_2) = f_2$ and $d(\u_1,\u_2) = 2$. We show the existence of such information vectors in Appendix~\ref{app:proof_claim1}. Given the result of Claim~\ref{claim:distance_ub}, we know that the entries of $\D_{\mathrm{mm}}$, $[\D_{\mathrm{mm}}]_{ij} = 2t + 1 - d_{\minmax_w}(f_i,f_j)$, are bounded from below by $2t-1$. The remaining part is to count the number of entries that satisfy $[\D_{\mathrm{mm}}]_{ij} = 2t$, i.e., the number of values $i,j$ for which $d_{\minmax_w}(f_i,f_j)=1$. We show that the number of such entries is equal to $4w(w-1)(w-2) + 2(w-1)$ by counting the number of function values that satisfy $d_{\minmax_w}(f_i,f_j) = 1$.

\begin{claim}\label{claim:distance_profile}
Consider a min-max function as defined in Definition~\ref{def:minmax}. For all $w\geq 3$ and $\ell \geq 3$, given a function value $f_1 = (i,j)$, the number of function values $f_2 \neq (i,j)$ that satisfy $d_{\minmax_w}(f_1,f_2) = 1$ is $4(w-2)$.
Therefore, the number of entries in $\D_{\mathrm{mm}}$ that is equal to $2t$ is equal to $4w(w-1)(w-2).$
\end{claim}
The proof of Claim~\ref{claim:distance_profile} consists of finding for every function value $f_1$ the number of distinct function values $f_2$ that can be obtained by changing one bit in any information vector $\u$ satisfying $\minmax_w(\u) = f_1$. A formal proof is provided in Appendix~\ref{app:proof_claim2}. The results of Claim~\ref{claim:distance_ub} and Claim~\ref{claim:distance_profile} characterize the entries of the function distance matrix $\D_{\mathrm{mm}}$. Recall that Theorem~\ref{thm:achievability} implies that
\begin{equation*}
    r_{\minmax_w}(k,t)\leq N\left(\D_{\mathrm{mm}}\right). 
\end{equation*}

We use Lemma~\ref{lemma:irregular:distance:code:gv} and the results of Claim~\ref{claim:distance_ub} and Claim~\ref{claim:distance_profile} to prove Lemma~\ref{lemma:redandancy_upper_bound_mm}. From Lemma~\ref{lemma:irregular:distance:code:gv} and by symmetry of $\D_{\mathrm{mm}}$ we have
\begin{equation*}
    N\left( \D_{\mathrm{mm}}\right) \leq  {\min_{r\in \mathbb{N}}}~\text{s.t.}~\Phi(r)\geq 0,
\end{equation*}
where $$\Phi'(r)=2^{r} - \underset{i \in [w(w-1)]}{\max} \sum_{j=1}^{i-1} V\left(r,[\D_{\mathrm{mm}}]_{\pi(i)\pi(j)}-1\right)$$ and $\pi$ is a permutation of the integers in $[w(w-1)]$. Note that $\sum_{j=1}^{i-1} V(r,[\D_{\mathrm{mm}}]_{\pi(i)\pi(j)}-1)$ is summing all the entries of a given row $\pi(i)$ of $\D_{\mathrm{mm}}$. Thus the maximum of this sum can be bounded from above by setting $i=w(w-1)$ and choosing a row with the largest entries.

From Claim~\ref{claim:distance_ub} and Claim~\ref{claim:distance_profile}, we know that a row $i$ with maximum entries contains exactly one entry equal to $0$, $4w-8$ entries equal to $2t$ and the rest is equal to $2t-1$. Given this observation, we obtain that $\Phi(r)$ in the Lemma statement is a lower bound to $\Phi'(r)$ and the lemma follows.
\end{proof}

We give an \fcc\ based on the single-parity check code in Construction~\ref{cons:minmax}.

\begin{construction}\label{cons:minmax}
Let $\mathcal{C}_\mathrm{SP}$ be a subcode of the single-parity check code of size $w(w-1)$. Replicate every bit in the codewords of $\mathcal{C}_\mathrm{SP}$ to $t$ bits. Assign a unique codeword of the expanded version of $\mathcal{C}_\mathrm{SP}$ to a redundancy vector $\p_{i,j}$ used for all information vectors $\u$ such that $f(\u)=(i,j)$.
\end{construction}

\begin{lemma}\label{lemma:minmax}
Construction~\ref{cons:minmax} is an \fcc\ for the min-max function and has redundancy %
		$ r_\mathrm{SP}= t(\left\lceil \log\left(w(w-1)\right)\right\rceil + 1). $
\end{lemma}

\begin{proof}
The lemma follows from the following observations: \begin{enumerate*}\item The length of each codeword in $\mathcal{C}_{\mathrm{SP}}$ is $\left\lceil \log\left(w(w-1)\right)\right\rceil + 1$; \item the minimum distance of $\mathcal{C}_{\mathrm{SP}}$ is $2$; and \item replicating every bit in the codewords of $\mathcal{C}_{\mathrm{SP}}$ gives the desired code of length $t\left(\left\lceil \log\left(w(w-1)\right)\right\rceil + 1\right)$, cardinality $w(w-1)$ and minimum distance $2t$.
\end{enumerate*}
\end{proof}

We present another {\fcc} based on Reed-Muller codes in Construction~\ref{cons:RMminmax}. For more information about Reed-Muller codes, we refer the reader to \cite{roth_introduction_2006}.

\begin{construction}\label{cons:RMminmax}
Consider the $\mathrm{RM}(r,m)$ Reed-Muller code of length $2^m$, cardinality $k_{\textrm{r,m}}\triangleq\sum_{i=0}^r \binom{m}{i},$ and minimum distance $2^{m-r}$. For given $w,t$ choose $m$ such that it is the smallest integer possible for which there exists an integer $r$ satisfying $2^{m-r}\geq 2t$ and $k_\textrm{r,m}\geq \log(w(w-1))$. Denote by $\p_{1,2},\p_{1,3},\dots,\p_{w-1,w}$ an arbitrary subcode of size $w(w-1)$ of the $\mathrm{RM}(r,m)$ code. We then define
$$ \Enc_{\minmax_w}(\u) = (\u, \p_{\minmax_w(\u)}). $$
\end{construction}

Following the arguments of Lemma~\ref{lemma:minmax}, it is clear that Construction~\ref{cons:RMminmax} gives an FCC for the min-max function with redundancy $r_\mathrm{RM} = 2^m$. To see the importance of this construction, consider the example where $t$ is a power of $2$ and $w\leq \sqrt{8t}$. Then, one can use an $\mathrm{RM}(1,\log (4t))$ to obtain an FCC for the min-max function with redundancy equal to $4t$ which is asymptotically, for large $w$, only $3$ bits away from the lower bound of Lemma~\ref{lemma:redandancy_lower_bound_mm}.

\section{Real-Valued Functions}\label{sec:ml_functions}
\begin{figure*}[t!]
	\begin{minipage}{0.2\textwidth}
{
\begin{tikzpicture}
\begin{axis}[
	width = 4cm,
	height = 4cm,
    axis lines = left,
    xlabel = $x$,
    ylabel = {$g(x)$},
    title = {$\mathrm{Sigmoid}$\vphantom{()}},
    grid = major,
]
\addplot [
    domain=-10:10, 
    samples=100, 
    color=black,
]
{1/(1+e^(-x))};
\end{axis}
\end{tikzpicture}%
}%
\end{minipage}%
\begin{minipage}{0.2\textwidth}
\begin{tikzpicture}
\begin{axis}[
	width = 4cm,
	height = 4cm,
    axis lines = left,
    xlabel = $x$,
    title = {$\tanh(x)$},
    grid = major,
]
\addplot [
    domain=-10:10, 
    samples=100, 
    color=black,
    ]
    {(e^x - e^(-x))/(e^x+e^(-x))};
\end{axis}
\end{tikzpicture}%
\end{minipage}%
\begin{minipage}{0.2\textwidth}
\begin{tikzpicture}
\begin{axis}[
	width = 4cm,
	height = 4cm,
    axis lines = left,
    xlabel = $x$,
    title = {$\mathrm{ReLU}$\vphantom{()}},
    grid = major,
]
\addplot [
    domain=-10:10, 
    samples=100, 
    color=black,
    ]
    {max(0,x)};
\end{axis}
\end{tikzpicture}%
\end{minipage}%
\begin{minipage}{0.2\textwidth}
\begin{tikzpicture}
\begin{axis}[
	width = 4cm,
	height = 4cm,
    axis lines = left,
    xlabel = $x$,
    title = {$\mathrm{Der.~of~sigmoid}$\vphantom{()}},
    grid = major,
]
\addplot [
    domain=-10:10, 
    samples=100, 
    color=black,
    ]
    {e^x/(1+e^(x))^2};
\end{axis}
\end{tikzpicture}%
\end{minipage}%
\begin{minipage}{0.2\textwidth}
\begin{tikzpicture}
\begin{axis}[
	width = 4cm,
	height = 4cm,
    axis lines = left,
    xlabel = $x$,
    title = {$\mathrm{Der.~of}$ $\tanh(x)$},
    grid = major,
    ymax = 1,
]
\addplot [
    domain=-10:10, 
    samples=100, 
    color=black,
    ]
    {1 - tanh(x)^2};
\end{axis}
\end{tikzpicture}
\end{minipage}%
	\caption{Plots of the considered real-valued functions. From left to right we plot the sigmoid function, hyperbolic tangent function, the ReLU, the derivative of the sigmoid function and the derivative of the hyperbolic tangent. Observe that the first three functions are bijective on a certain interval of numbers and are approximately constant on one or two intervals. The last two functions are symmetric around $0$. For all $x\geq 0$, the last two functions are bijective on a certain interval and $0$ otherwise.}
	\label{fig:ml_functions}
\end{figure*}

In this section we apply our theoretical results on \fccs\ to a collection of real-valued functions that take a real number as input and output a real number, i.e., functions of the form $g:\mathbb{R}\to \mathbb{R}$. Throughout this work, however, we consider digital functions that take binary vectors as input and have an arbitrary output, i.e., we consider functions of the form $f: \mathbb{Z}_2^k \to \Im(f)$. To this end, let $\bintoreal:\Z_2^k \to \mathbb{R}$ be a mapping from the binary information vectors to real numbers. Thus, throughout this section considered functions are\footnote{From a practical point of view, we assume that the data is stored in binary vectors with a predetermined representation and a desired precision. However, during computations, the input is transformed to a real number and the function is computed over the reals with different representation and different precision.}
$$ \ML_g(\u) = g(\bintoreal(\u)), $$
where $g$ is one of the functions presented below. While our ideas apply to several binary representations, we opt to explain the results using a fixed precision quantization $\bintoreal$ as follows. Given a fixed precision $\epsilon>0$, the mapping $\bintoreal$ maps the binary vectors to intervals of size $\epsilon$, i.e.,
$$ \bintoreal(\u) = \epsilon  \left(\mathrm{bin2dec}(\u)-2^{k-1}+0.5 \right),$$
where $\mathrm{bin2dec}: \Z_2^k\to \{0,1,\dots2^k-1\}$ is the standard mapping from binary to decimal. Notice that this way, real values in the range of $\pm (2^{k-1}-0.5)\epsilon$ can be represented, which will be called \emph{quantization intervals} hereafter.

The functions we consider are shown in Fig.~\ref{fig:ml_functions} and are defined as follows:
\begin{itemize}
	\item Sigmoid or logistic function: $\sigma(x) = \frac{1}{1+e^{-x}}$ and its derivative $\frac{\partial\sigma(x)}{\partial x} = \frac{e^x}{(1+e^{x})^2}$.
    \item Rectified linear unit (ReLU) function: $\mathrm{ReLU}(x) = \max\{0,x\}$ and its derivative $\frac{\partial \mathrm{ReLU}(x)}{\partial x} = 0,$ if $x<0$ and $1$, if $x>0$.
    \item Hyperbolic tangent function: $\tanh(x) = \frac{e^x - e^{-x}}{e^x+e^{-x}}$ and its derivative $ \frac{\partial\tanh(x)}{\partial x}=1 - \tanh(x)^2.$
\end{itemize}

Those functions have practical importance as they are activation functions, and their derivatives, used in neural networks\footnote{In neural networks, the considered functions are multivariate, they take as input a model vector $\mathbf{a}$ and a weight vector $\mathbf{w}$. However, they only operate on the inner product $\mathbf{w}^T\mathbf{a}$. To that end, we express the inner product of those vectors by the scalar $x = \mathbf{w}^T\mathbf{a}$ and treat those functions as univariate.}.

Throughout this section we discuss \fccs~whose encoding is based on the function value only, as in Theorem \ref{thm:achievability}. Thus, the defining quantity of interest is the function distance matrix $\D_{\ML_g}(t,f_1,\dots,f_\I)$, which we abbreviate by $\D_{g}$.

We now study the three functions mentioned above. We delay the study of the derivatives of the sigmoid and $\tanh(x)$ function for only after Lemma~\ref{lem:ml_functions2}. These three functions can be divided into two classes: a class of functions that are bijective on a certain interval and constant (output equal to $0$) otherwise, such as the $\mathrm{ReLU}$ function; and a class of functions that are bijective on a certain interval and have approximately constant output for small and large values of $x$, such as the $\mathrm{sigmoid}$ and $\tanh$ functions. To see this division notice that numerically one can consider $\tanh(x) = 1$ for all $x\geq 6$ and $\tanh(x) = - 1$ for all $x\leq - 6$. Similarly $\sigma(x) = 1$ for $x\geq 10$ and $\sigma(x) = 0$ for $x\leq -10$. The ReLU function is a bijective function for all $x> 0$ and is $0$ otherwise (cf. Fig.~\ref{fig:ml_functions}).

Let $[a,b]\subset \mathbb{R}$ be the interval in which the function $g$ is bijective and assume for simplicity that $\epsilon$ divides $b-a$. For notational convenience, we denote the binary vector representing a certain quantization center $c$ by $\mathbf{w}_i$, $\mathbf{u}_i$ or $\mathbf{v}_i$ if $c<a$, $a\leq c\leq b$ or $c>b$, respectively. In addition we define $\displaystyle d(\u_i,\mathbf{v}) \triangleq \min_{\ell} d(\u_i,\mathbf{v}_\ell)$ to be the Hamming distance between the binary vector $\u_i$ representing a quantization center $c_1\in [a,b]$ and all binary vectors representing a quantization center $c_2<a$. We define $\displaystyle d(\u_i,\mathbf{v})$ and $d(\mathbf{v}, \mathbf{w})$ similarly.

We characterize the redundancy of an \fcc\ for the considered real-valued functions in Lemma~\ref{lem:ml_functions1} and Lemma~\ref{lem:ml_functions2}. %
Let $g_{00}:\mathbb{R}\to \mathbb{R}$ be a real-valued function that is bijective on an interval $[a,b] \subset \mathbb{R}$ and equal to $0$ on $\mathbb{R}\setminus[a,b]$. Fix an $\epsilon>0$, and
define the \emph{symmetric} square matrix $\D_{\ML\mathrm{00}}$ with $(b-a)/\epsilon + 1$ rows that has for $i\leq j$ the entries%
\begin{equation*}
    [\D_{\ML\mathrm{00}}]_{ij} = \left\{ \!\! \begin{array}{ll}
	0, & \!\! \text{if } i=j,\\
	2t\!+\!1\!-\!d(\u_i,\u_j),& \!\! \text{if } j \leq \frac{b-a}{\epsilon},\\
	2t\!+\!1\!-\! \min \{d(\u_i,\mathbf{v}),d(\u_i,\mathbf{w})\},& \!\! \text{otherwise.}
	\end{array} \right.
\end{equation*}

\begin{lemma}\label{lem:ml_functions1}
The redundancy of an \fcc\ for the function $\ML_{g_{00}}$ is bounded from above by
$$r_{g_{00}}(k,t) \leq N\left( \D_{\ML\mathrm{00}}\right).$$
\end{lemma}

If only one input vector evaluates to $0$, the upper bound becomes the optimal redundancy of an \fcc\ for this function. The same holds if several input vectors evaluate to $0$ and have similar distance profiles to each of the $\u_i$'s. This observation holds for the next lemma as well.

\begin{proof}

The proof follows from Theorem~\ref{thm:achievability} by designing the parities based on the function values. On a high level, since all input values in $\mathbb{R}\setminus [a,b]$ have the same output value, then the codewords of the form $(\mathbf{w}_i,\p)$ and $(\mathbf{v}_j,\p)$ are allowed to be confusable after $t$ errors and can thus have a distance less than $2t+1$. However, the codewords of the form $(\u_i,\p)$ cannot be confusable with any other codeword after $t$ errors. Therefore, for every $\u_i$ we search for the closest (in Hamming distance) $\mathbf{v}$ or $\mathbf{w}$ and design the parity vector of the $\mathbf{v}_j$'s and $\mathbf{w}_j$'s accordingly. The same is done for $\u_i$ and $\u_j$ for $i\neq j$.

Formally, let $\p_0,\dots,\p_{\frac{b-a}{\epsilon}}$ be the parity vectors used in the encoding such that $\Enc(\u_i) = (\u_i, \p_i)$ for $i= 1,\dots, \frac{b-a}{\epsilon}$, $\Enc(\mathbf{w}) = (\mathbf{w}, \p_0)$ and $\Enc(\mathbf{v}) = (\mathbf{v}, \p_0)$. It follows from Theorem~\ref{thm:achievability} that $r_{g_{00}}(k,t) \leq N\left(\D_{\ML\mathrm{00}}\right).$
\end{proof}

Let $g_{01}:\mathbb{R}\to \mathbb{R}$ be a real-valued function that is bijective on an interval $[a,b] \subset \mathbb{R}$ and satisfies $g_{01}(x) = 0$ for all $x< a$ and $g_{01}(x) = 1$ for all $x>b$. Fix a precision $\epsilon$ and, for ease of notation, define $\mathbf{v} \triangleq \u_{\frac{b-a}{\varepsilon}+1}$ and $\mathbf{w} \triangleq \u_{\frac{b-a}{\varepsilon}+2}$. Let the symmetric matrix $\D_{\ML\mathrm{01}}$ with $(b-a)/\epsilon + 2$ rows be defined as follows 
	$$[\D_{\ML\mathrm{01}}]_{ij} \! = \! \left\{ \!\! \begin{array}{ll}
	0, & \text{if } i=j,\\
	2t\!+\!1\!-\!d(\u_i,\u_j),& \text{otherwise.}
	\end{array} \right.$$
\begin{lemma}
\label{lem:ml_functions2}
The redundancy of an \fcc\ for the function $\ML_{g_{01}}(\u) = g_{01}(\bintoreal(\u))$ is bounded from above by
$$r_{g_{01}}(k,t) \leq N\left( \D_{\ML\mathrm{01}}\right).$$
\end{lemma}

The proof is omitted as it follows the same steps of the proof of Lemma~\ref{lem:ml_functions1} while taking care of not confusing any of the $\mathbf{v}_i$'s with any of the $\mathbf{w}_i$'s.

Now we study the derivative of $\sigma(x)$ and $\tanh(x)$. Both functions are symmetric around $0$ and are bijective on an interval $[0,a]$ and constant otherwise. %
Numerically, one could consider the derivative of $\sigma(x)$ to be equal to $0$ outside the interval $[-10,10]$ and the derivative of $\tanh(x)$ to be $0$ outside the interval $[-6,6]$.

For this set of functions we abuse notation and denote by $\u_i$ the binary representation of a quantization center $c \in [0,a]$ and by $-\u_i$ the binary representation of the quantization center $-c$. Similarly $\mathbf{v}_i$ is the binary representation of $c>a$ and $-\mathbf{v}_i$ is the binary representation of $-c<-a$. We define $d(\pm\u_i,\pm\u_j)\triangleq \min\{d(\u_i,\u_j),d(\u_i,-\u_j),d(-\u_i,\u_j),d(-\u_i,-\u_j)\}$ and define $d(\pm\u_i,\pm\mathbf{v})$ similarly. This notation makes the following definitions easier to present.

Let $g_\mathrm{sym}:\mathbb{R} \to \mathbb{R}$ be a real-valued function with $g(x)=g(-x)$, bijective on an interval $[0,a] \subset \mathbb{R}$ and is constant for $x>a$. Fix a precision $\epsilon$ and define the \emph{symmetric} square matrix $\D_{\ML\mathrm{-sym}}$ with $\frac{a}{\epsilon} + 1$ rows that has for $i\leq j$ the entries 
	$$ [\D_{\ML\mathrm{-sym}}]_{ij} \!=\! \left\{ \!\! \begin{array}{ll}
	0, & \text{if } i=j,\\
	2t\!+\!1\!-\!d(\pm\u_i,\pm\u_j),& \text{if } j \leq \frac{a}{\epsilon},\\
	2t\!+\!1\!-\! d(\pm\u_i,\pm\mathbf{v}),& \text{otherwise.}
	\end{array} \right.$$

\begin{lemma}\label{lem:ml_der_functions}
The redundancy of an \fcc\ for the function $\ML_{g_{\mathrm{sym}}}(\u) = g_{\mathrm{sym}}(\bintoreal(\u))$ is then bounded from above by
$$r_{g_\mathrm{sym}} \leq N\left(\D_{\ML\mathrm{-sym}}\right).$$
\end{lemma}
The proof is omitted as it follows the same steps of the proof of Lemma~\ref{lem:ml_functions1}.

\section{Conclusion}\label{sec:conclusion}
We introduced a new class of codes called function-correcting codes which encode a message to allow a successful recovery of a certain attribute or a function value of this message after transmission over an erroneous channel. This encoding potentially reduces the redundancy compared to error-correcting codes by leveraging the side information given to the receiver by the knowledge of the possibly erroneous original message and the desired function.

We considered an encoding setup in which the message itself is also transmitted and restricted our attention to substitution channels with at most $t$ errors. For this setting, we derived lower and upper bounds on the redundancy of \fccs\ by establishing a connection to irregular-distance codes. Further, we examined several functions of interest for which we derived explicit distance matrices, such that an irregular-distance code satisfying the distance matrix gives an optimal \fcc\ for the function at hand. Furthermore, we derived lower bounds and constructed \fccs\ for each specific function. Our constructions have optimal redundancy for the Hamming weight distribution functions. 
For the min-max function, we construct almost optimal codes. For the Hamming weight function there is still a gap of roughly $\frac23 t$ between the lower bound and the provided construction, leaving the problem of finding optimal \fccs\ open. For real-valued functions, a rigorous study of the distance profile of the input vectors is needed to understand the gap between the achievable redundancy and the lower bound. Further research directions on this topic include the study of \fccs\ for other functions of interest and under different channels.

\appendices
\section{Derivations of Redundancies in Table \ref{tab:redundancy}} \label{app:derivation:redundancies:table}
We start by deriving the redundancy obtained by employing a standard  error-correcting onto the data, labeled as the column ``\emph{ECC on Data}'' in Table \ref{tab:redundancy}. This means, the data vector $\u$ is encoded with a systematic code of dimension $k$ and minimum distance $2t+1$. The redundancy part $\p$ of this systematic code is then appended to $\u$, resulting in $(\u,\p)$. Clearly, with such a construction it is possible to reconstruct $\u$ at the receiver and thus $f(\u)$. It is known \cite[Ch. 5.5]{roth_introduction_2006} that there exists a binary alternant code of length $n$, minimum distance $2t+1$ and redundancy at most $r \leq t\lceil\log n \rceil$. Since $n = k+r$,
\begin{align*}
	r &\leq t\lceil\log (k+r) \rceil \leq t \log(k+r) +t \\
	&= t \log k +t\log (1+r/k) +t \\
	& \leq t \log k + t \frac{r \log \mathrm{e}}{k} + t.
\end{align*}
It follows that
$$ r \leq \frac{t\log k+t}{ (1-t/k\log \mathrm{e})} $$
and thus, for large $k$ and fixed $t$, the dominant term is $t \log k$.

We now turn to derive the redundancy obtained by a direct approach of encoding the function values, which corresponds to the column ``\emph{ECC on Function Values}'' in Table~\ref{tab:redundancy}. More precisely, we encode the function value $f(\u)$ with a (possibly non-systematic) code of cardinality $E$ (recall that $E$ is the size of the image of $f$) and minimum distance $2t+1$. The resulting \emph{codeword} $\c$ is then appended to $\u$, resulting in $(\u,\c)$. Also in this case, it is possible to retrieve $f(\u)$ by decoding the function value from the received word corresponding to $\c$ and simply ignoring the information part $\u$. In this case, the redundancy of our construction is given by the length of the employed code. Using alternant codes, we obtain for the redundancy of the alternant code
\begin{align*}
	r_{\mathrm{alt}} &\leq  t \lceil\log (\log \lceil|E|\rceil + r_{\mathrm{alt}})\rceil \\
	&\leq t \log\log|E| +t + t(1+r_{\mathrm{alt}})\log \mathrm{e} /\log |E|
\end{align*}
and thus 
$$ r_{\mathrm{alt}} \leq \frac{t \log \log |E|+t(1+\log \mathrm{e})}{1-t/\log |E| \log \mathrm{e}}. $$
Consequently, the redundancy of the direct approach is given by the length of the alternating code $r = \lceil\log|E|\rceil+r_{\mathrm{alt}}$. For sufficiently large $|E|$ and fixed $t$ this is adequately approximated by $\log |E| + t\log \log |E|$.

\section{Proof of Lemma~\ref{lemma:regular:distance:code}}\label{app:proof_lemma4}
\begin{proof}[Proof of Lemma~\ref{lemma:regular:distance:code}]
	Lemma \ref{lemma:irregular:distance:code:gv} states that there exists a code of cardinality $M$, minimum distance $D$ and length $r$, if $2^r > MV(r,D-1)$. For $D-1 \leq r/2$, we can use \cite[Lemma 4.7.2]{ash_information_1990} to bound the size of the Hamming ball to $V(r,D-1) \leq 2^r \mathrm{e}^{-2r(\frac12-\frac{D-1}{r})^2}$. Combining these two results, we obtain that if $2^r> M 2^r \mathrm{e}^{-2r(\frac12-\frac{D-1}{r})^2}$, then there exists an $[M,D]$ code of length $r$. Setting $D = r/2-\epsilon r$ for some $0<\epsilon\leq\frac12$, we can deduce that there exists an $[M,D]$ code of length $r$ satisfying $M \leq \mathrm{e}^{2r\epsilon^2}$. Choosing $\epsilon = \sqrt{\ln(r)/r}$, we obtain that $r = 2D/(1-2\sqrt{\ln(r)/r})$. Here we require $r\geq 10$ such that $\epsilon\leq\frac12$. We can then use that $\ln(D)/D \geq \ln(r)/r$ for $r\geq D \geq 3$ and we obtain the lemma's statement.
\end{proof}

\section{Proof of Claim~\ref{claim:distance_ub}}\label{app:proof_claim1}
\begin{proof}[Proof of Claim~\ref{claim:distance_ub}]
	We give a proof for $\ell = 3$. For $\ell >3$, we can restrict all the bits of all $\u^{(v)}$, $v\in[w]$ to be $0$ except for the three least significant bits and apply the same proof of $\ell = 3$.
	We show that for all $i,j,i',j'\in [w]$, $(i,j)\neq(i',j')$ there exist two information words $\u,\u'$ such that $\minmax_w(\u) = (i,j)$ and $ \minmax_w(\u')= (i',j')$, where $d(\u,\u') = 2$. We split the proof into the following three cases.
	
	\begin{itemize}
		\item {\em $i'\leq i$: }To change $\u$ into $\u'$ satisfying $\minmax_w(\u)=(i,j)$ and $\minmax_w(\u')=(i',j')$, consider $\u$ to be of the form%
		\begin{equation*}\label{eq:ipli}
			\u = (001,\dots,\underbrace{000}_{\u^{(i)}}, 001,\dots,\underbrace{010}_{\u^{(j)}}, 001,\dots,001).
		\end{equation*}
		Note that $\minmax_w(\u)=(i,j)$ by definition of $\minmax_w$. We can change $\u$ to $\u'$ as follows. First, if $i'<i$ flip the third bit of $\u^{(i')}$, (so that $\uii{i'} = (000)$) to change the function value to $(i',j)$. To change $j$ to $j'$, it is sufficient to flip the first bit of $\u^{(j')}$. Thus, $d_{\minmax_w}((i,j),(i',j'))\leq 2$ because we could edit $\u$ with $\minmax_w(\u) = (i,j)$ to $\u'$ with $\minmax_w(\u') = (i',j')$ using only two substitutions.
		
		\item {\em $i'>i$, $i'\neq j$: }Consider $\u$ to be of the form%
		\begin{equation*}
			\u = (001,\dots,\underbrace{000}_{\u^{(i)}}, 001,\dots,\underbrace{000}_{\uii{i'}},\underbrace{010}_{\u^{(j)}}, 001,\dots,001).
		\end{equation*}
		Note that $\minmax_w(\u)=(i,j)$ by definition of $\minmax_w$. We can change $\u$ to $\u'$ as follows. First flip the third bit of $\u^{(i)}$, (so that $\ui = (001)$) to change the function value to $(i',j)$. To change $j$ to $j'$, it is sufficient to flip the first bit of $\u^{(j')}$.
		
		\item {\em $i'>i$ and $i'=j$: }Consider $\u$ to be of the form%
		\begin{equation}\label{eq:ipej}
			\u = (010,\dots,\underbrace{001}_{\u^{(i)}}, 010,\dots,\underbrace{100}_{\u^{(j)}}, 010,\dots,010).
		\end{equation}
		Note that $\minmax_w(\u)=(i,j)$ by definition of $\minmax_w$. We can change $\u$ to $\u'$ as follows. Flip the first bit of $\u^{(j)}$, (so that $\uj = (000)$) to change the function value to $(j,1)=(i',1)$ (or $(j,2)$, if $j=1$). To obtain $j'$ as the maximum, it is sufficient to flip the first bit of $\u^{(j')}$.
	\end{itemize}
\end{proof}

\section{Proof of Claim~\ref{claim:distance_profile}}\label{app:proof_claim2}
\begin{proof}[Proof of Claim~\ref{claim:distance_profile}]
	We give a proof for $\ell=3$. For $\ell >3$, we can restrict all the bits of all $\u^{(v)}$, $v \in [w]$ to be $0$ except for the three least significant bits and apply the same proof of $\ell = 3$. %
	Fix $f_1 \triangleq (i,j)$ and consider all information words $\u$ such that $\minmax_w(\u) = f_1$. Note that for any $\u$, the $\uii{v}$'s form a totally ordered set and therefore can be arranged in a chain, as illustrated in Fig. \ref{fig:proof:claim}.
	By definition, for any $f_2$ with $d_{\minmax_w}(f_1,f_2) = 1$, there exists a $\u$ with $\minmax_w(\u)=f_1$, such that, by flipping one bit in $\u$, the function value changes from $f_1$ to $f_2$. We find all possible function values that can be obtained after a single bit flip in some $\u$ with $\minmax_w(\u)=(i,j$). We distinguish between the following types of edit operations.
	\begin{enumerate}
		\item Change one bit in $\ui$. First, change $\ui$ such that the result becomes larger then $\ui$, but smaller than $\uj$. This way it is only possible to change the function value to $(v,j)$, for an arbitrary $v \in [w] \setminus \{i,j\}$. This can in fact be achieved by choosing $\u$ to be
		$$ \u = (011,\dots,\underbrace{001}_{\u^{(i)}}, \underbrace{010}_{\uv},011,\dots,\underbrace{111}_{\u^{(j)}}, 011,\dots,011),$$
		and flipping the first bit of $\ui$ (so that $\ui = (101)$). Note that $\minmax_w(\u) = (i,j)$.
		
		Second, change $\ui$ such that it becomes larger than $\uj$. This way, it is only possible to change the function value to $(v,i)$, $v \in [w] \setminus\{i,j\}$. This can be achieved by choosing $\u$ to be
		$$ \u = (011,\dots,\underbrace{001}_{\u^{(i)}}, \underbrace{010}_{\uv},011,\dots,\underbrace{100}_{\u^{(j)}}, 011,\dots,011)$$
		and flipping the first bit of $\ui$ (so that $\ui = (101)$). For an illustration, see Fig.~\ref{fig:proof:claim}. 
		
		\item Change one bit in $\uj$. First, we change $\uj$ such that the result becomes smaller then $\uj$, but larger than $\ui$. This way it is only possible to change the function value to $(i,v)$, for an arbitrary $v \in [w] \setminus \{i,j\}$. This can in fact be achieved by choosing $\u$ to be
		$$ \u = (100,\dots,\underbrace{000}_{\u^{(i)}}, \underbrace{101}_{\uv},100,\dots,\underbrace{110}_{\u^{(j)}}, 100,\dots,100)$$
		and flipping the first bit of $\uj$ (so that $\uj = (010)$). Note that $\minmax_w(\u) = (i,j)$.
		
		Second, we change $\uj$ such that it becomes smaller than $\ui$. This way, is is only possible to change the function value to $(j,v)$, $v \in [w] \setminus\{i,j\}$. This can be achieved by choosing $\u$ to be
		$$ \u = (100,\dots,\underbrace{011}_{\u^{(i)}}, \underbrace{101}_{\uv},100,\dots,\underbrace{110}_{\u^{(j)}}, 100,\dots,100)$$
		and flipping the first bit of $\uj$ (so that $\uj = (010)$).
		\item Change one bit in $\uii{v}$, $v \in [w] \setminus\{i,j\}$. This does not yield any additional function values that can be reached, since it is only possible to obtain $(v,j)$ or $(i,v)$.
	\end{enumerate}
	\begin{figure}
		\centering
		\begin{tikzpicture}
			\node (cap) {1) $(i,j) \rightarrow (v,j)$};
			\node[right= 1cm of cap] (order) {$\ui < \uii{v}<\dots<\uj$};
			\draw[->] ($(order.south) + (-1.5,0)$) to [bend right=35] (order.south);
			
			\node[below=0.2cm of cap] (cap2) {\hphantom{2) }$(i,j) \rightarrow (v,i)$};
			\node[right= 1cm of cap2] (order) {$\ui < \uii{v}<\dots<\uj$};
			\draw[->] ($(order.south) + (-1.5,0)$) to [bend right=20] ($(order.south) + (1.75,0)$);
			
			\node[below=0.4cm of cap2] (cap3) {2) $(i,j) \rightarrow (i,v)$};
			\node[right= 1cm of cap3] (order) {$\ui <\dots< \uii{v}<\uj$};
			\draw[->] ($(order.south) + (1.5,0)$) to [bend left=35] (order.south);
			
			\node[below=0.2cm of cap3] (cap4) {\hphantom{2) }$(i,j) \rightarrow (j,v)$};
			\node[right= 1cm of cap4] (order) {$\ui <\dots< \uii{v}<\uj$};
			\draw[->, bend right] ($(order.south) + (1.5,0)$) to [bend left=20] ($(order.south) + (-1.75,0)$);
		\end{tikzpicture}
		\caption{Illustration of the different editing operations in the proof of Claim \ref{claim:distance_profile}.}
		\label{fig:proof:claim}
	\end{figure}
	Since the resulting function values in cases 1) and 2) are distinct, for each $f_1$, there exist $4(w-2)$ values $f_2$ with $f_1\neq f_2$ and $d_{\minmax_w}(f_1,f_2)$.
	Using further, that there are $w(w-1)$ function values, the total number of entries in $\D_{\minmax}$ that are equal to $2t$ is equal to
	$4w(w-1)(w-2)$.
\end{proof}

\bibliographystyle{ieeetr}
\bibliography{IEEEabrv,FunctionalErrorCorrection}

\end{document}